\documentclass[a4paper, UKenglish, cleveref, numberwithinsect]{lipics-arxiv}

\hideLIPIcs

\usepackage{xspace}
\usepackage[lined,commentsnumbered]{algorithm2e}
\DontPrintSemicolon
\usepackage{mathtools, bm}
\usepackage{xargs}
\usepackage{tikz}
\usetikzlibrary{snakes, positioning}
\newtheorem{hypothesis}[theorem]{Hypothesis}

\newcommand*{\ie}{i.e.\@\xspace}

\newcommand{\set}[1]{\{#1\}} 
\newcommand{\integers}[0]{\mathbb{Z}} 
\newcommand{\Z}{\integers}
\newcommand{\naturals}[0]{\mathbb{N}} 

\newcommand{\N}{\naturals}
\newcommand{\Uu}{\mathcal{U}}
\newcommand{\Vv}{\mathcal{V}}
\newcommand{\Ww}{\mathcal{W}}
\newcommand{\class}[1]{\textsf{#1}} 
\newcommand{\poly}[1]{\textsf{poly}(#1)} 

\newcommand{\Oh}{\mathcal{O}}

\newcommand{\abs}[1]{\lvert #1 \rvert}

\newcommand{\len}[1]{\mathrm{len}(#1)} 
\newcommand{\eff}[1]{\mathrm{eff}(#1)} 
\renewcommand{\vec}[1]{{\bf #1}}

\newcommand{\configuration}[3]{#1(\vec{#2}#3)} 
\newcommand{\config}[2]{\configuration{#1}{#2}{}} 
\newcommand{\covfig}[2]{\configuration{#1}{#2}{'}} 

\newcommand{\run}[3]{#1\xrightarrow{#2}#3} 

\newcommand{\PreserveBackslash}[1]{\let\temp=\\#1\let\\=\temp}
\newcolumntype{C}[1]{>{\PreserveBackslash\centering}p{#1}}
\newcolumntype{R}[1]{>{\PreserveBackslash\raggedleft}p{#1}}
\newcolumntype{L}[1]{>{\PreserveBackslash\raggedright}p{#1}}

\newcommand{\problemx}[3]{
\par\noindent\underline{\sc#1}\par\nobreak\vskip.2\baselineskip
\begingroup\clubpenalty10000\widowpenalty10000
\setbox0\hbox{\bf INPUT:\ }\setbox1\hbox{\bf QUESTION:\ }
\dimen0=\wd0\ifnum\wd1>\dimen0\dimen0=\wd1\fi
\vskip-\parskip\noindent
\hbox to\dimen0{\box0\hfil}\hangindent\dimen0\hangafter1\ignorespaces#2\par
\vskip-\parskip\noindent
\hbox to\dimen0{\box1\hfil}\hangindent\dimen0\hangafter1\ignorespaces#3\par
\endgroup}

\definecolor{lightpink}{RGB}{253, 164, 229}
\definecolor{lightblue}{RGB}{171, 236, 255}
\definecolor{lightorange}{RGB}{255, 204, 144}

\newcommand\norm[1]{\lVert#1\rVert}

\newcommand{\thin}{\mathrm{thin}}
\newcommand{\tail}{\mathrm{tail}}

\newcommand{\trans}[1]{\xrightarrow{#1}}              
\newcommand{\reach}{\trans{*}}                       

\newcommand{\Mod}[1]{\,(\mathrm{mod}\ #1)}

\newcommand{\Multiply}[2]{\textbf{\texttt{Multiply}}$[#1, #2]$}
\newcommand{\MultiplyGadget}{\textbf{\texttt{Multiply}}\xspace}
\newcommand{\Divide}[2]{\textbf{\texttt{Divide}}$[#1, #2]$}
\newcommand{\DivideGadget}{\textbf{\texttt{Divide}}\xspace}
\newcommand{\Edge}[1]{\textbf{\texttt{Edge}}$[#1]$}
\newcommand{\EdgeGadget}{\textbf{\texttt{Edge}}\xspace}
\newcommand{\HyperEdge}[1]{\textbf{\texttt{HyperEdge}}$[#1]$}
\newcommand{\HyperEdgeGadget}{\textbf{\texttt{HyperEdge}}\xspace}
\newcommand{\Vertices}[1]{\textbf{\texttt{Verticies}}$[#1]$}

\newcommand{\VertexSelected}[1]{\textbf{\texttt{VertexSelected}}$[#1]$}
\newcommand{\VertexSelectedGadget}{\textbf{\texttt{VertexSelected}}\xspace}
\newcommand{\increase}[2]{$#1\,+${}$=#2$}
\newcommand{\decrease}[2]{$#1\,-${}$=#2$}

\newcommand{\vr}[1]{\mathsf{#1}}

\title{Coverability in VASS Revisited: Improving Rackoff’s Bound to Obtain Conditional Optimality}

\titlerunning{Coverability in VASS Revisited} 

\author{Marvin K\"unnemann}
{RPTU Kaiserslautern-Landau}
{kuennemann@cs.uni-kl.de}
{}
{Research partially supported by the Deutsche Forschungsgemeinschaft (DFG, German Research Foundation) -- 462679611.}
 
\author{Filip Mazowiecki}
{University of Warsaw, Poland}
{f.mazowiecki@mimuw.edu.pl}
{}
{Supported by the ERC grant INFSYS, agreement no. 950398.}

\author{Lia Sch\"utze}
{Max Planck Institute for Software Systems (MPI-SWS)}
{lschuetze@mpi-sws.org}
{0000-0003-4002-5491}
{}

\author{Henry Sinclair-Banks}
{Centre for Discrete Mathematics and its Applications (DIMAP) \& Department of Computer Science, University of Warwick, Coventry, UK \and \url{http://henry.sinclair-banks.com}}
{h.sinclair-banks@warwick.ac.uk}
{https://orcid.org/0000-0003-1653-4069}
{Supported by EPSRC Standard Research Studentship (DTP), grant number EP/T5179X/1.}

\author{Karol W\k{e}grzycki}
{Saarland University and Max Planck Institute for Informatics, Saarbr\"ucken, Germany}
{wegrzycki@cs.uni-saarland.de}
{https://orcid.org/0000-0001-9746-5733}
{Supported by the project TIPEA that has received funding from the European Research Council (ERC) under the European Unions Horizon 2020 research and innovation programme (grant agreement No. 850979).}

\authorrunning{M. K\"{u}nnemann, F. Mazowiecki, L. Sch\"{u}tze, H. Sinclair-Banks, and K. W\k{e}grzycki}



\acknowledgements{
    We would like to thank our anonymous reviewers for their comments, for their time, and especially for highlighting a piece of related work~\cite{KoppenhagenM00}.
}

\nolinenumbers

\begin{document}

\maketitle

\begin{abstract}
    Seminal results establish that the coverability problem for Vector Addition Systems with States (VASS) is in \class{EXPSPACE} (Rackoff, '78) and is \class{EXPSPACE}-hard already under unary encodings (Lipton, '76).
More precisely, Rosier and Yen later utilise Rackoff's bounding technique to show that if coverability holds then there is a run of length at most $n^{2^{\Oh(d \log d)}}$, where $d$ is the dimension and $n$ is the size of the given unary VASS.
Earlier, Lipton showed that there exist instances of coverability in $d$-dimensional unary VASS that are only witnessed by runs of length at least $n^{2^{\Omega(d)}}$.
Our first result closes this gap.
We improve the upper bound by removing the twice-exponentiated $\log(d)$ factor, thus matching Lipton's lower bound.
This closes the corresponding gap for the exact space required to decide coverability.
This also yields a deterministic $n^{2^{\Oh(d)}}$-time algorithm for coverability.
Our second result is a matching lower bound, that there does not exist a deterministic $n^{2^{o(d)}}$-time algorithm, conditioned upon the Exponential Time Hypothesis.

When analysing coverability, a standard proof technique is to consider VASS with bounded counters.
Bounded VASS make for an interesting and popular model due to strong connections with timed automata.
Withal, we study a natural setting where the counter bound is linear in the size of the VASS.
Here the trivial exhaustive search algorithm runs in $\Oh(n^{d+1})$-time.
We give evidence to this being near-optimal.
We prove that in dimension one this trivial algorithm is conditionally optimal, by showing that $n^{2-o(1)}$-time is required under the $k$-cycle hypothesis.
In general fixed dimension $d$, we show that $n^{d-2-o(1)}$-time is required under the 3-uniform hyperclique hypothesis.
    \begin{picture}(0,0)
        \put(-90,-240) {\hbox{\includegraphics[width=40px]{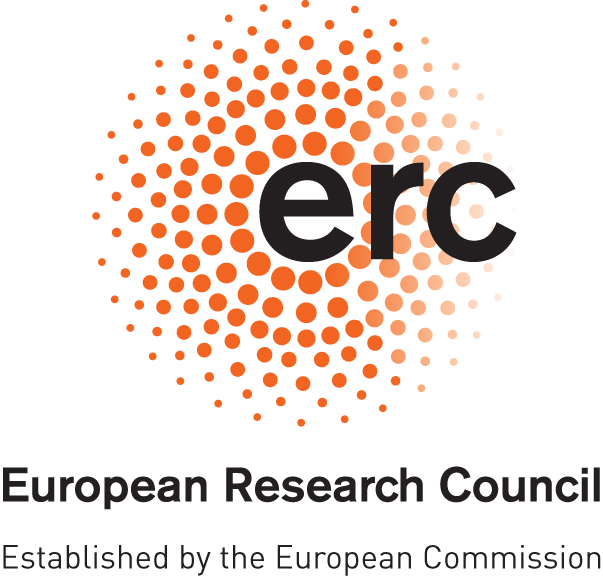}}}
        \put(-45,-260) {\hbox{\includegraphics[width=60px]{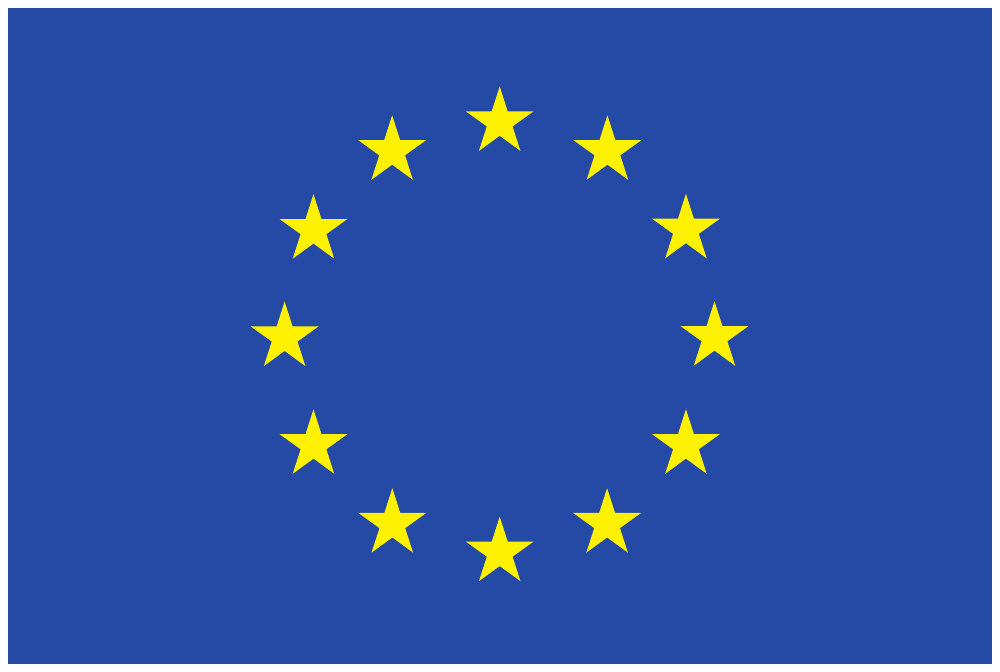}}}
    \end{picture}
\end{abstract}

\section{Introduction}
\label{sec:introduction}

Vector Addition Systems with States (VASS) are a popular model of concurrency with a number of applications in database theory~\cite{BojanczykDMSS11}, business processes~\cite{Aalst97}, and more (see the survey~\cite{Schmitz16}).
A $d$-dimensional VASS ($d$-VASS) consists of a finite automaton equipped with $d$ non-negative valued counters that can be updated by transitions.
A configuration in a $d$-VASS consists of a state and a $d$-dimensional vector over the naturals.
One of the central decision problems for VASS is the \emph{coverability problem}, that asks whether there is a run from a given initial configuration to some configuration with at least the counter values of a given target configuration. 
Coverability finds application in the verification of safety conditions, which often equate to whether or not a particular state can be reached without any precise counter values~\cite{ComonJ98, GantyM12}.
Roughly speaking, one can use VASS as a modest model for concurrent systems where the dimension corresponds with the number of locations a process can be in and each counter value corresponds with the number of processes in a particular location~\cite{EsparzaLMMN14,GermanS92}.

In 1978, Rackoff~\cite{Rackoff78} showed that coverability is in \class{EXPSPACE}, by proving that if coverability holds then there exists a run of double-exponential length.
Following, Rosier and Yen~\cite{RosierY86} analysed and discussed Rackoff's ideas in more detail and argued that if a  coverability holds then it is witnessed by a run of length at most $n^{2^{\Oh(d \log d)}}$, where $n$ is the size of the given unary encoded $d$-VASS.
Furthermore, this yields a $2^{\Oh(d\log{d})}\cdot\log(n)$-space algorithm for coverability.
Prior to this in 1976, Lipton~\cite{Lipton76} proved that coverability is \class{EXPSPACE}-hard even when VASS is encoded in unary, by constructing an instance of coverability witnessed only by a run of double-exponential length $n^{2^{\Omega(d)}}$.
Rosier and Yen~\cite{RosierY86} also presented a proof that generalises Lipton's constructions to show that $2^{\Omega(d)}\cdot\log(n)$-space is required for coverability.
Although this problem is \class{EXPSPACE}-complete in terms of classical complexity, a gap was left open for the exact space needed for coverability~{\cite[Section 1]{RosierY86}}.
By using an approach akin to Rackoff's argument, we close this thirty-eight-year-old gap by improving the upper bound to match Lipton's lower bound.

\begin{enumerate}[Result 1:]
	\item If coverability holds then there exists a run of length at most $n^{2^{\Oh(d)}}$ (Theorem~\ref{thm:min-run}).
	Accordingly, we obtain an optimal $2^{\Oh(d)}\cdot\log(n)$-space algorithm that decides coverability (Corollary~\ref{cor:coverability-algorithm}).
\end{enumerate}

Our bound also implies the existence of a deterministic $n^{2^{\Oh(d)}}$-time algorithm for coverability.
We complement this with a matching lower bound on the deterministic running time that is conditioned upon the Exponential Time Hypothesis (ETH).

\begin{enumerate}[Result 2:]
	\item Under ETH, there is no deterministic $n^{2^{o(d)}}$-time algorithm deciding coverability in unary $d$-VASS (Theorem~\ref{thm:unary-kvass-coverability-lb}).
\end{enumerate}

While our results establish a fast-increasing, conditionally optimal exponent of $2^{\Theta(d)}$ in the time complexity of the coverability problem, they rely on careful constructions that enforce the observation of large counter values.
In certain settings, however, it is natural to instead consider a restricted version of coverability, where all counter values remain \emph{bounded}. This yields one of the simplest models, fixed-dimension bounded unary VASS, for which we obtain even tighter results.
Decision problems for $B$-bounded VASS, where $B$ forms part of the input, have been studied due to their strong connections to timed automata~\cite{HaaseOW12,FearnleyJ13,MazowieckiP19}.
We consider linearly-bounded unary VASS, that is when the maximum counter value is bounded above by a constant multiple of the size of the VASS.
Interestingly, coverability and reachability are equivalent in linearly-bounded unary VASS.
The trivial algorithm that employs depth-first search on the space of configurations runs in $\Oh(n^{d+1})$-time for both coverability and reachability.
We provide evidence that the trivial algorithm is optimal.

\begin{enumerate}[Result 3:]
	\item Reachability in linearly-bounded unary 1-VASS requires $n^{2-o(1)}$-time, subject to the $k$-cycle hypothesis (Theorem~\ref{thm:k-cycle-to-vass}).
\end{enumerate}

This effectively demonstrates that the trivial algorithm is optimal in the one-dimensional case.
For the case of large dimensions, we show that the trivial algorithm only differs from an optimal deterministic-time algorithm by at most an $n^{3+o(1)}$-time factor.

\begin{enumerate}[Result 4:]
	\item Reachability in linearly-bounded unary $d$-VASS requires $n^{d-2-o(1)}$-time, subject to the 3-uniform $k$-hyperclique hypothesis (Theorem~\ref{thm:bounded-reachability}).
\end{enumerate}

Broadly speaking, these results add a time complexity perspective to the already known space complexity, that is for any fixed dimension $d$, coverability in unary $d$-VASS is \class{NL}-complete~\cite{Rackoff78}.

\subparagraph*{Organisation and Overview}
Section~\ref{sec:space-upper-bound} contains our first main result, the improved upper bound on the space required for coverability.
Most notably, in Theorem~\ref{thm:min-run} we show that if coverability holds then there exists a run of length at most $n^{2^{\Oh(d)}}$.
Then, in Corollary~\ref{cor:coverability-algorithm} we are able to obtain a non-deterministic $2^{\Oh(d)}\cdot\log(n)$-space algorithm and a deterministic $n^{2^{\Oh(d)}}$-time algorithm for coverability.
In much of the same way as Rackoff, we proceed by induction on the dimension.
The difference is in the inductive step; Rackoff's inductive hypothesis dealt with a case where all counters are bounded by the same well-chosen value.
Intuitively speaking, the configurations are bounded within a $d$-hypercube. This turns out to be suboptimal. 
This is due to the fact that the volume of a $d$-hypercube with sides of length $\ell$ is $\ell^d$; unrolling the induction steps gives a bound of roughly $n^{d \cdot (d-1) \cdot \ldots \cdot 1} = n^{d!} = n^{2^{\Oh(d \log d)}}$, hence the twice-exponentiated $\log(d)$ factor.
The key ingredient in our proof is to replace the $d$-hypercubes with a collection of objects with greatly reduced volume, thus reducing the number of configurations in a run witnessing coverability.

Section~\ref{sec:time-lower-bound} contains our second main result, the matching lower bound on the time required for coverability that is conditioned upon ETH.
In Lemma~\ref{lem:k-clique-reduction}, we first reduce from finding a $k$-clique in a graph to an instance of coverability in bounded unary 2-VASS with zero-tests.
Then, via Lemma~\ref{lem:rosier-yen}, we implement the aforementioned technique of Rosier and Yen to, when there is a counter bound, remove the zero-tests at the cost of increasing to a $d$-dimensional unary VASS.
Then, in Theorem~\ref{thm:unary-kvass-coverability-lb} we are able to conclude, by setting $k=2^d$, that if ETH holds, then there is no deterministic $n^{2^{o(d)}}$-time algorithms for coverability in unary $d$-VASS.
This is because ETH implies that there is no $f(k) \cdot n^{o(k)}$-time algorithm for finding a $k$-clique in a graph with $n$ vertices (Theorem~\ref{thm:k-clique}).

Section~\ref{sec:bounded-vass} contains our other results where we study bounded fixed dimension unary VASS.
Firstly, Theorem~\ref{thm:k-cycle-to-vass} states that under the $k$-cycle hypothesis (Hypothesis~\ref{hyp:k-cycle-hypothesis}), there does not exist a deterministic $n^{2-o(1)}$-time algorithm deciding reachability in linearly-bounded unary 1-VASS.
Further, we conclude in Corollary~\ref{cor:unary-2vass-coverability}, if the $k$-cycle hypothesis is assumed then there does not exist a deterministic $n^{2-o(1)}$-time algorithm for coverability in (not bounded) unary 2-VASS.
Following, we prove Theorem~\ref{thm:bounded-reachability}, that claims there does not exist a deterministic $n^{d-o(1)}$-time algorithm reachability in linearly-bounded unary $(d+2)$-VASS under the 3-uniform $k$-hyperclique hypothesis (Hypothesis~\ref{hyp:hyperclique}).
We achieve this with two components.
First, in Lemma~\ref{lem:hyperclique-reduction}, we first reduce from finding a $4d$-hyperclique to an instance of reachability in a bounded unary $(d+1)$-VASS with a fixed number of zero-tests.
Second, via Lemma~\ref{lem:czerwinski-orlikowski}, we implement the newly developed controlling counter technique of Czerwi\'nski and Orlikowski~\cite{CzerwinskiO21} to remove the fixed number of zero-tests at the cost of increasing the dimension by one.

\subparagraph{Related Work}
The coverability problem for VASS has plenty of structure that still receives active attention. 
The set of configurations from which the target can be covered is upwards-closed, meaning that coverability still holds if the initial counter values are increased.
An alternative approach, the \emph{backwards algorithm} for coverability, relies on this phenomenon.
Starting from the target configuration, one computes the set of configurations from which it can be covered~\cite{AbdullaCJT00}. 
Thanks to the upwards-closed property, it suffices to maintain the collection of minimal configurations. 
The backwards algorithm terminates due to Dickson's lemma, however, using Rackoff's bound one can show it runs in double-exponential time~\cite{BozzelliG11}. 
This technique has been deeply analysed for coverability in VASS and some extensions~\cite{FigueiraFSS11,LazicS21}. 
Despite high complexity, there are many implementations of coverability relying on the backwards algorithm that work well in practice. 
Intuitively, the idea is to prune the set of configurations, using relaxations that can be efficiently implemented in SMT solvers~\cite{EsparzaLMMN14,BlondinFHH16,BlondinHO21}. 

Another central decision problem for VASS is the \emph{reachability problem}, asking whether there is a run from a given initial configuration to a given target configuration.
Reachability is a provably harder problem. In essence, reachability differs from coverability by allowing one zero-test to each counter.
Counter machines, well-known to be equivalent to Turing machines~\cite{Minsky67}, can be seen as VASS with the ability to arbitrarily zero-test counters; coverability and reachability are equivalent here and are undecidable.
In 1981, Mayr proved that reachability in VASS is decidable~\cite{Mayr81}, making VASS one of the richest decidable variants of counter machines.
Only recently, after decades of work, has the complexity of reachability in VASS been determined to be Ackermann-complete~\cite{LerouxS19,CzerwinskiO21,Leroux21}.
A widespread technique for obtaining lower bounds for coverability and reachability problems in VASS is to simulate counter machines with some restrictions.
Our overall approach to obtaining lower bounds follows suit; we first reduce finding cliques in graphs, finding cycles in graphs, and finding hypercliques in hypergraphs to various intermediate instances of coverability in VASS with extra properties such as bounded counters or a fixed number of zero-tests.
These VASS, that are counter machines restricted in some way, are then simulated by standard higher-dimensional VASS.
Such simulations are brought about by the two previously developed techniques.
Rosier and Yen leverage Lipton's construction to obtain VASS that can simulate counter machines with bounded counters~\cite{RosierY86}.
Czerwi\'nski and Orlikowski have shown that the presence of an additional counter in a VASS, with carefully chosen transition effects and reachability condition, can be used to implicitly perform a limited number of zero-tests~\cite{CzerwinskiO21}.

Recently, some work has been dedicated to the coverability problem for low-dimensional VASS~\cite{AlmagorCPSW20,MazowickiSW23}.
Furthermore, reachability in low-dimensional VASS has been given plenty of attention, in particular for 1-VASS~\cite{ValiantP75,HaaseKOW09} and for 2-VASS~\cite{HopcroftP79,BlondinEFGHLMT21}.
In the restricted class of flat VASS, other fixed dimensions have also been studied~\cite{Czerwinski0LLM20,CzerwinskiO22}.

Another studied variant, \emph{bidirected} VASS, has the property that for every transition $(p, \vec{x}, q)$, the reverse transition $(q, -\vec{x}, p)$ is also present.
The reachability problem in bidirected VASS is equivalent to the uniform word problem in commutative semigroups, both of which are \class{EXPSPACE}-complete~\cite{MayrM82}; not to be confused with the reversible reachability problem in general VASS which is also \class{EXPSPACE}-complete~\cite{Leroux13}.
In 1982, Meyer and Mayr listed an open problem that stated, in terms of commutative semigroups, the best known upper bound for coverability in general VASS~\cite{Rackoff78}, the best known lower bound for coverability in bidirected VASS~\cite{Lipton76}, and asked for improvements to these bounds~{\cite[Section 8, Problem 3]{MayrM82}}.
Subsequently, Rosier and Yen refined the upper bound for coverability in general VASS to $2^{\Oh(d\log{d})}\cdot\log(n)$-space~\cite{RosierY86}.
Finally, Koppenhagen and Mayr showed that the coverability problem in bidirected VASS can be decided in $2^{\Oh(n)}$-space~\cite{KoppenhagenM00}, matching the lower bound.

\section{Preliminaries}
\label{sec:preliminaries}

We use bold font for vectors.
We index the $i$-th component of a vector $\vec{v}$ by writing $\vec{v}[i]$.
Given two vectors $\vec{u}, \vec{v} \in \Z^d$ we write $\vec{u} \leq \vec{v}$ if $\vec{u}[i] \leq \vec{v}[i]$ for each $1 \leq i \leq d$. 
For every $1 \leq i \leq d$, we write $\vec{e}_i \in \Z^d$ to represent the $i$-th standard basis vector that has $\vec{e}_i[i] = 1$ and $\vec{e}_i[j] = 0$ for all $j \neq i$.
Given a vector $\vec{v} \in \Z^d$ we define $\norm{\vec{v}} = \max\set{1, \abs{\vec{v}[1]}, \ldots, \abs{\vec{v}[d]}}$.
Throughout, we assume that $\log$ has base $2$.
We use $\poly{n}$ to denote $n^{\Oh(1)}$.

A \emph{$d$-dimensional Vector Addition System with States} ($d$-VASS) $\Vv = (Q, T)$ consists of a non-empty finite set of states $Q$ and a non-empty set of transitions $T \subseteq Q \times \Z^d \times Q$.
A \emph{configuration} of a $d$-VASS is a pair $(q, \vec{v}) \in Q \times \N^d$ consisting of the current state $q$ and current counter values $\vec{v}$, denoted $q(\vec{v})$.
Given two configurations $\config{p}{u}$, $\config{q}{v}$, we write $\config{p}{u} \trans{} \config{q}{v}$ if there exists $t = (p, \vec{x}, q) \in T$ where $\vec{x} = \vec{v} - \vec{u}$. 
We may refer to $\vec{x}$ as the \emph{update} of a transition and  may also write $\config{p}{v} \trans{t} \config{q}{w}$ to emphasise the transition $t$ taken.

A \emph{path} in a VASS is a (possibly empty) sequence of transitions $((p_1, \vec{x}_1, q_1), \ldots, (p_\ell, \vec{x}_\ell, q_\ell))$, where $(p_i, \vec{x}_i, q_i) \in T$ for all $1 \leq i \leq \ell$ and such that the start and end states of consecutive transitions match $q_i = p_{i+1}$ for all $1 \leq i \leq \ell-1$.
A \emph{run} $\pi$ in a VASS is a sequence of configurations $\pi = (\configuration{q_0}{v}{_0}, \ldots, \configuration{q_\ell}{v}{_\ell})$ such that $\configuration{q_i}{v}{_i} \trans{} \configuration{q_{i+1}}{v}{_{i+1}}$ for all $1 \leq i \leq \ell-1$. 
We denote the length of the run by $\len{\pi} = \ell+1$. 
If there is such a run $\pi$, we can write $\configuration{q_0}{v}{_0} \trans{\pi} \configuration{q_\ell}{v}{_\ell}$.
We may also write $\config{p}{u} \reach \config{q}{v}$ if there exists a run from $\config{p}{u}$ to $\config{q}{v}$.
The \emph{underlying path} of a run $\pi$ is sequence of transitions $(t_1, \ldots, t_\ell)$ taken between each of the configurations in $\pi$, so $\configuration{q_i}{v}{_i} \trans{t_{i+1}} \configuration{q_{i+1}}{v}{_{i+1}}$ for all $0 \leq i \leq \ell-1$. 

A \emph{$B$-bounded} $d$-VASS, in short $(B, d)$-VASS, is given as an integer upper bound on the counter values $B \in \N$ and $d$-VASS $\Vv$.
A configuration in a $(B, d)$-VASS is a pair $\config{q}{v} \in Q \times \{0, \ldots, B\}^d$. 
The notions of paths and runs in bounded VASS remain the same as for VASS, but are accordingly adapted for the appropriate bounded configurations.
We note that one should think that $B$ forms part of the problem statement,  not the input, as it will be given implicitly by a function depending on the size of the VASS.
For example, we later consider \emph{linearly-bounded} $d$-VASS, that represent occasions where $B = \Oh(\norm{\Vv})$.

We do allow for zero-dimensional VASS, that is VASS with no counters, which can be seen as just directed graphs.
A \emph{hypergraph} is a generalisation of the graph. 
Formally, a hypergraph is a tuple $H = (V, E)$ where $V$ is a set of vertices and $E$ is a collection of non-empty subsets of $V$ called \emph{hyperedges}. 
For an integer $\mu$, a hypergraph is \emph{$\mu$-uniform} if each hyperedge has cardinality $\mu$. 
Note that a 2-uniform hypergraph is a standard graph.

We study the complexity of the \emph{coverability problem}. 
An instance $(\Vv, \config{p}{u}, \config{q}{v})$ of coverability asks whether there is a run in the given VASS $\Vv$ from the given initial configuration $\config{p}{u}$ to a configuration $\covfig{q}{v}$ with at least the counter values $\vec{v}' \geq \vec{v}$ of the given target configuration $\config{q}{v}$.
At times, we also consider the \emph{reachability problem} that additionally requires $\vec{v}' = \vec{v}$ so that the target configuration is reached exactly.

To measure the complexity of these problems we need to discuss the encoding used.
In \emph{unary encoding}, a $d$-VASS $\Vv = (Q, T)$ has \emph{size} $\norm{\Vv} = \abs{Q} + \sum_{(p, \vec{x}, q) \in T} \norm{\vec{x}}$.
We define a \emph{unary} $d$-VASS $\Uu = (Q', T')$ to have restricted transitions $T' \subseteq Q' \times \set{-1, 0, 1}^d \times Q'$, the size is therefore $\norm{\Uu} = \abs{Q'} + \abs{T'}$.
For any unary encoded $d$-VASS $\Vv$ there exists an equivalent unary $d$-VASS $\Uu$ such that $\norm{\Uu} = \norm{\Vv}$.
An instance $(\Vv, \config{p}{u}, \config{q}{v})$ of coverability has size $n = \norm{\Vv} + \norm{\vec{u}} + \norm{\vec{v}}$.
An equal in size, equivalent instance $(\Vv', p'(\vec{0}), q'(\vec{0}))$ of coverability exists; consider adding an initial transition $(p', \vec{s}, p)$ and a final transition $(q, -\vec{t}, q')$.

It is well known that for $d$-VASS, the coverability problem can be reduced to the reachability problem. 
Indeed, for an instance $(\Vv, \config{p}{u}, \config{q}{v})$ of coverability, define $\Vv' = (Q, T')$ that has additional decremental transitions at the target states $T' = T \cup \set{(q, \vec{e}_i, q) : 1 \leq i \leq d}$.
It is clear that $\config{p}{u} \reach \covfig{q}{v}$, for some $\vec{v}' \geq \vec{v}$, in $\Vv$ if and only if $\config{p}{u} \reach \config{q}{v}$ in $\Vv'$. 
\begin{lemma}[folklore]\label{lem:cov_to_reach}
    Let $(\Vv$, $\config{p}{u}$, $\config{q}{v})$ be an instance of coverability. 
    It can be reduced to an instance of reachability $(\Vv'$, $\config{p}{u}$, $\config{q}{v})$ such that $\norm{\Vv'} = \Oh(\norm{\Vv})$.
\end{lemma}

A \emph{$d$-dimensional Vector Addition System} ($d$-VAS) $\Vv$ is a system without states, consisting only of a non-empty collection of transitions $\Vv \subseteq \Z^d$.
All definitions, notations, and problems carry over for VAS except that, for simplicity, we drop the states across the board. 
For example, a configuration in a VAS is just a vector $\vec{v} \in \N^d$.
Another well-known result from the seventies by Hopcroft and Pansiot, one can simulate the states of a VASS at the cost of three extra dimensions in a VAS~\cite{HopcroftP79}.
For clarity, the VAS obtained has an equivalent reachability relation between configurations; a configuration $q(\vec{x})$ in the original VASS corresponds with a configuration $(\vec{x}, a, b, c)$ in the VAS, where $a$, $b$, and $c$ represent the state $q$.
\begin{lemma}[{\cite[Lemma 2.1]{HopcroftP79}}]\label{lem:no_states}
    A $d$-VASS $\Vv$ can be simulated by $(d+3)$-VAS
    $\mathcal{\Vv'}$ such that $\norm{\mathcal{\Vv'}} = \poly{\norm{\Vv}}$.
\end{lemma}

\section{Improved Bounds on the Maximum Counter Value}
\label{sec:space-upper-bound}

This section is devoted to our improvement of the seminal result of Rackoff.
Throughout, we fix our attention to the arbitrary instance $(\Vv, \config{p}{s}, \config{q}{t})$ of the coverability problem in a $d$-VASS $\Vv = (Q,T)$ from the initial configuration $\config{p}{s}$ to a configuration $\covfig{q}{t}$ with at least the counter values of the target configuration $\config{q}{t}$.
We denote $n = \norm{\Vv} + \norm{\vec{s}} + \norm{\vec{t}}$. 
Informally, $n$ may as well be the number of states plus the absolute value of the greatest update on any transition, for these differences can be subsumed by the second exponent in our following upper bounds. 
The following two theorems follow from Rackoff's technique and subsequent work by Rosier and Yen, in particular see~{\cite[Lemma 3.4 and Theorem 3.5]{Rackoff78}} and~{\cite[Theorem 2.1 and Lemma 2.2]{RosierY86}}.

\begin{theorem}[Corollary of~{\cite[Lemma 3.4]{Rackoff78}} and~{\cite[Theorem 2.1]{RosierY86}}]\label{thm:rackoff-runs}
    Suppose $\config{p}{s} \reach \covfig{q}{t}$ for some $\vec{t}' \ge \vec{t}$.
    Then there exists a run $\pi$ such that $\config{p}{s} \trans{\pi} \config{q}{t''}$ for some $\vec{t}'' \geq \vec{t}$ and $\len{\pi} \leq n^{2^{\Oh(d\log d)}}$. 
\end{theorem}

\begin{theorem}[cf.~{\cite[Theorem 3.5]{Rackoff78}}]
    \label{thm:rackoff-algorithm}
	For a given $d$-VASS $\Vv$, integer $\ell$, and two configurations $p(\vec{s})$ and $q(\vec{t})$, there is an algorithm that determines the existence of a run $\pi$ of length $\len{\pi} \leq \ell$ that witnesses coverability, so $\config{p}{s} \trans{\pi} \covfig{q}{t}$ for some $\vec{t}' \geq \vec{t}$. 
	The algorithm can be implemented to run in non-deterministic $\Oh(d\log(n\cdot\ell))$-space or deterministic $2^{\Oh(d\log(n\cdot\ell))}$-time.
\end{theorem}
\begin{proof}
	In runs whose length is bounded by $\ell$, the observed counter values are trivially bounded by $n\cdot\ell$.
	Notice that every configuration can be written in $\Oh(d \log(n\cdot\ell))$ space.
    A non-deterministic algorithm can therefore decide coverability by guessing a path on-the-fly by only maintaining the current configuration.
    The algorithm accepts if and only if $\vec{t}$ is covered by the final configuration.

	The second part follows from the standard construction that if a problem can be solved in $S(n)$ non-deterministic space then it can be solved in $2^{\Oh(S(n))}$ deterministic time. 
    Indeed, one can construct the graph of all configurations and check whether there is a path from the initial configuration to the final configuration. 
    Since there are at most $2^{\Oh(d \log(n\cdot\ell))}$ many configurations, this can be completed in $2^{\Oh(d \log(n\cdot\ell))}$-time. 
\end{proof}

Note that Theorem~\ref{thm:rackoff-runs} combined with Theorem~\ref{thm:rackoff-algorithm} yield non-deterministic $2^{\Oh(d\log
d)}$\nobreakdash-space and deterministic $n^{2^{\Oh(d\log(d))}}$-time algorithms for coverability. 
Our result improves this by a $\Oh(\log(d))$ factor in the second exponent.

\begin{theorem}\label{thm:min-run}
    Suppose $\config{p}{s} \reach \covfig{q}{t}$ for some $\vec{t}' \ge \vec{t}$.
    Then there exists a run $\pi$ such that $\config{p}{s} \trans{\pi} \config{q}{t''}$ for some $\vec{t}'' \geq \vec{t}$ and $\len{\pi} \leq n^{2^{\Oh(d)}}$. 
\end{theorem}

This combined with Theorem~\ref{thm:rackoff-algorithm} yields the following corollary.
\begin{corollary}\label{cor:coverability-algorithm}
    Coverability in $d$-VASS can be decided by both a non-deterministic $2^{\Oh(d)} \cdot \log(n)$-space algorithm and a deterministic $n^{2^{\Oh(d)}}$-time algorithm.
\end{corollary}

Note that by Lemma~\ref{lem:no_states}, we may handle VAS instead of VASS. 
Recall that, as there are no states, a $d$-VAS consists only of a set of vectors in $\Z^d$ that we still refer to as transitions.
A configuration is just a vector in $\N^d$.
Accordingly, we may fix our attention on the instance $(\Vv, \vec{s}, \vec{t})$ of the coverability problem in a $d$-VAS $\Vv = \set{\vec{v}_1, \ldots, \vec{v}_m}$ from the initial configuration $\vec{s}$ to a configuration $\vec{t}'$ that is at least as great as the target configuration $\vec{t}$.
The rest of this section is dedicated to the proof of Theorem~\ref{thm:min-run}. 
Imitating Rackoff's proof, we proceed by induction on the dimension $d$. 
Formally, we prove a stronger statement; Theorem~\ref{thm:min-run} is a direct corollary of the following lemma.

\begin{lemma}\label{lem:main_induction}
    Define $L_i \coloneqq n^{4^{i}}$, and let $\vec{t} \in \N^d$ such that $\norm{\vec{t}} \leq n$. 
    For any $\vec{s} \in \N^d$, if $\vec{s} \reach \vec{t}'$ for some $\vec{t}' \geq \vec{t}$ then there exists a run $\pi$ such that $\vec{s} \trans{\pi} \vec{t}''$ for some $\vec{t}'' \ge \vec{t}$ and $\len{\pi} \le L_d$.
\end{lemma}

The base case is $d=0$.
In a $0$-dimensional VAS, the only possible configuration is the empty vector $\bm{\varepsilon}$ and therefore there is only the trivial run $\bm{\varepsilon} \reach \bm{\varepsilon}$.
This trivially satisfies the lemma.

For the inductive step, when $d \geq 1$, we assume that Lemma~\ref{lem:main_induction} holds for all lower dimensions $0, \ldots, d-1$.
Let $\pi = (\vec{c}_0, \vec{c}_1, \ldots, \vec{c}_\ell)$ be a run with minimal length such that $\run{\vec{s}}{\pi}{\vec{t}'}$ for some $\vec{t}' \geq \vec{t}$, so in particular, $\vec{c}_0 = \vec{s}$ and $\vec{c}_n = \vec{t}'$. 
Our objective is to prove that $\len{\pi} = \ell+1 \leq L_d$. 
Observe that configurations $\vec{c}_i$ need to be distinct, else $\pi$ could be shortened trivially. 
We introduce the notion of a \emph{thin configuration}.
\begin{definition}[Thin Configuration]
    In a $d$-VAS, we say that a configuration $\vec{c} \in \N^d$ is
    \emph{thin} if there exists a permutation $\sigma$ of $\set{1, \ldots, d}$ such that $\vec{c}[\sigma(i)] < M_i$ for every $i \in \{1,\ldots,d\}$,
    where $M_0 \coloneqq n$ and for $i \geq 1$, $M_i \coloneqq L_{i-1} \cdot n$.
\end{definition}

Recall, from above, the run $\pi = (\vec{c}_0, \vec{c}_1, \ldots, \vec{c}_\ell)$. 
Let $t \in \set{0, \ldots, \ell}$ be the first index where $\vec{c}_t$ is not thin, otherwise let $t = \ell+1$ if every configuration in $\pi$ is thin.
We decompose the run about the $t$-th configuration $\pi = \pi_\thin \cdot \pi_\tail$, where $\pi_\thin \coloneqq (\vec{c}_0, \ldots, \vec{c}_{t-1})$ and $\pi_\tail \coloneqq (\vec{c}_t, \ldots, \vec{c}_\ell)$.
Note that $\pi_\thin$ or $\pi_\tail$ can be empty.
Subsequently, we individually analyse the lengths of $\pi_\thin$ and $\pi_\tail$ (see Figure~\ref{fig:thin-config}).
We will also denote $\vec{m} = \vec{c}_t$ to be the first configuration that is not thin.

\begin{figure}[ht!]
    \centering
    \includegraphics[width=0.5\textwidth]{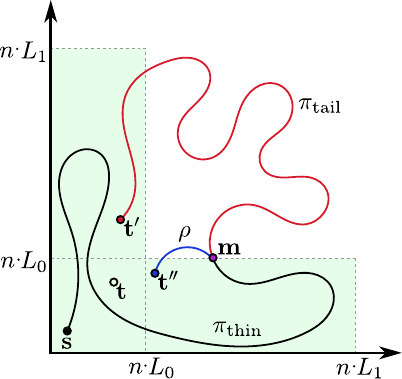}
    \caption{The schematic view of proofs of Claim~\ref{clm:thin} and Claim~\ref{clm:tail}, restricted to the two-dimensional case.
    Note that $\vec{s}$ is the initial configuration and $\vec{t}$ is the target configuration.
    Every configuration inside the green shaded polygon is thin, where each rectangular component of the green shaded polygon corresponds to a  permutation of the indices.
    Observe that $\vec{m}$ is the first configuration, just outside the green shaded polygon, that is not thin.
    Claim~\ref{clm:thin} bounds $\pi_\thin$, and therefore its maximum length, by the volume of the green polygon.
    Claim~\ref{clm:tail} argues that there is an executable run $\rho$ (drawn in blue) from $\vec{m}$ to $\vec{t}'' \geq \vec{t}$ of length at most $L_{d-1}$ that can be used in place of the run $\pi_\tail$ (drawn in red) from $\vec{m}$ to $\vec{t}' \geq \vec{t}$.}
    \label{fig:thin-config}
\end{figure}

\begin{claim}\label{clm:thin}
    $\len{\pi_\thin} \leq d! \cdot n^d \cdot L_{d-1} \cdot \ldots \cdot L_0$.
\end{claim}
\begin{proof}
    By definition, every configuration in $\pi_\thin$ is thin. 
    Moreover, since $\pi$ has a minimal length, no configurations in $\pi$ repeat, let alone in $\pi_\thin$.
    We now count the number of possible thin configurations.
    There are $d!$ many permutations of $\set{1, \ldots, d}$.
    For a given permutation $\sigma$ and an index $i \in \set{1, \ldots, d}$, we know that for a thin configuration $\vec{c}$, $0 \leq \vec{c}[\sigma(i)] < M_i$, so there are at most $M_i = L_{i-1} \cdot n$ many possible values on the $\sigma(i)$-th counter.
    Hence the total number of thin configurations is at most $d! \cdot \prod_{i=1}^d (L_{i-1} \cdot n) = d! \cdot n^d \cdot L_{d-1} \cdot \ldots \cdot L_0$.
\end{proof}

\begin{claim}\label{clm:tail}
    $\len{\pi_\tail} \le L_{d-1}.$
\end{claim}
\begin{proof}
    Consider $\vec{m} \in \N^d$, the first configuration of $\pi_\tail$.
    Let $\sigma$ be a permutation such that $\vec{m}[\sigma(1)] \leq \vec{m}[\sigma(2)] \leq \ldots \leq \vec{m}[\sigma(d)]$.
    Given that $\vec{m}$ is not thin, for every permutation $\sigma'$ there exists an $i \in \set{1, \ldots, d}$ such that $\vec{m}[\sigma'(i)] \ge M_i$; in particular, this holds for $\sigma$. Note that this also implies $M_i \leq \vec{m}[\sigma(i+1)] \leq \ldots \leq \vec{m}[\sigma(d)]$.

    We construct an $(i-1)$-VAS $\Uu$ from $\Vv$ by ignoring the counters $\sigma(i), \ldots, \sigma(d)$. 
    Formally, $\vec{u} \in \Uu$ if there is $\vec{v} \in \Vv$ such that $\vec{u}[j] = \vec{v}[\sigma(j)]$ for each $1 \leq j \leq i-1$.
    In such a case we say $\vec{u}$ is the \emph{projection} of $\vec{v}$ via $\sigma$. We will use the inductive hypothesis to show that there is a short path $\rho'$ in $\Uu$ from (the projection of) $\vec{m}$ covering (the projection of) $\vec{t}$. We will then show that the remaining components of $\vec{m}$ are large enough that the embedding of $\rho'$ into $\Vv$ maintains its covering status.

    Recall that $\vec{t}'$ is the final configuration of the run $\pi$.
    Note that the run $\pi_\tail$ induces a run $\pi_\tail'$ in $\Uu$ by
    permuting and projecting every configuration.
    More precisely, $(\vec{m}[\sigma(1)], \ldots, \vec{m}[\sigma(i-1)]) \trans{\pi_\tail'} (\vec{t}'[\sigma(1)], \ldots, \vec{t}'[\sigma(i-1)])$.
    By the inductive hypothesis there exists a run $\rho'$ in $\Uu$ such that $(\vec{m}[\sigma(1)], \ldots, \vec{m}[\sigma(i-1)]) \trans{\rho'} (\vec{t}''[\sigma(1)], \ldots, \vec{t}''[\sigma(i-1)])$, such that $(\vec{t}''[\sigma(1)], \ldots, \vec{t}''[\sigma(i-1)]) \geq (\vec{t}[\sigma(1)], \ldots, \vec{t}[\sigma(i-1)])$ and $\len{\rho'} \le L_{i-1}$.

    Let $(\vec{u}_1, \ldots, \vec{u}_{\len{\rho'}})$ be the underlying path of the run $\rho'$, that is, the sequence of transitions in $\Uu$ that are sequentially added to form the run $\rho'$.
    By construction, each transition vector $\vec{u}_i \in \Uu$ has a corresponding transition vector $\vec{v}_i \in \Vv$ where $\vec{u}_i$ is the projection of $\vec{v}_i$ via $\sigma$.
    We will now show that the following run witnesses coverability of $\vec{t}$. 
    \begin{equation*}
        \rho = \left( \vec{m}, \, \vec{m} + \vec{v}_1, \, \vec{m} + \vec{v}_1 + \vec{v}_2, \, \ldots, \, \vec{m} + \sum_{j = 1}^{\len{\rho'}} \vec{v}_j \right)
    \end{equation*}
    
	To this end, we verify that (i) $\rho$ is a run, that is, all configurations lie in $\N^d$, and (ii) the final configuration indeed covers $\vec{t}$. 
    For	components $\sigma(1)$, \ldots, $\sigma(i-1)$, this follows directly from
	the inductive hypothesis. For all other components we will show that
	\emph{all} configurations of $\rho$ are covering $\vec{t}$ in these components. This satisfies both (i) and (ii).
    
    Let $j$ be any of the remaining components. 
    Recall that by the choice of $\vec m$, $\vec{m}[j] \geq M_i = n \cdot L_{i-1}$. 
    Since $n > \norm{\Vv} \geq \norm{\vec{v}_j}$ for every $1 \leq j \leq \len{\rho'}$, this means that in a single step, the value of a counter can change by at most $n$. 
    Given that $\len{\rho} = \len{\rho'} \leq L_{i-1}$, the value on each of the remaining components must be at least $n$ for every configuration in $\rho$.
	In particular, observing that $\norm{\vec{t}} \leq n$, the final configuration of $\rho$ satisfies
    \begin{equation*}
        \vec{m} + \sum_{j = 1}^{\len{\rho'}} \vec{v}_j \geq \vec{t}.
    \end{equation*}

	Finally, observe that $\len \rho = \len{\rho'} \leq L_{i-1} \leq L_{d-1}$.
\end{proof}

To conclude this section, we show that Lemma~\ref{lem:main_induction} follows from Claim~\ref{clm:thin} and Claim~\ref{clm:tail}.

\begin{proof}[Proof of Lemma~\ref{lem:main_induction}]
    From Claim~\ref{clm:thin} and Claim~\ref{clm:tail},
    \begin{align*}
        \len{\pi} \leq \len{\pi_\thin} + \len{\pi_\tail} 
        & \, \leq d! \cdot n^d \cdot L_{d-1} \cdot \ldots \cdot L_0 + L_{d-1} \\
        & \, \leq 2 \cdot d! \cdot n^d \cdot L_{d-1} \cdot \ldots \cdot L_0.
    \end{align*}
    Recall that $n \geq 2$ and observe that $2 \cdot d! \cdot n^d \le n^{2^d}$. Hence,
    \begin{equation*}
        \len{\pi} \leq n^{2^{d}} \cdot L_{d-1} \cdot \ldots \cdot L_0.
    \end{equation*}
    Next, we use the definition of $L_i \coloneqq n^{4^{i}}$ to show
    \begin{equation*}
        \len{\pi} \leq n^{2^{d}} \cdot \prod_{i=0}^{d-1} n^{4^i} \leq n^{\left(2^{d} + \sum_{i=0}^{d-1} 4^i\right)} .
    \end{equation*}
    Finally, when $d \geq 1$, $2^{d} + \sum_{i=0}^{d-1} 4^i \leq 4^d$ holds, therefore
    \begin{equation*}
        \len{\pi} \le n^{4^d} = L_d \qedhere. 
    \end{equation*}
\end{proof}

\section{Conditional Time Lower Bound for Coverability}
\label{sec:time-lower-bound}

In this section, we present a conditional lower bound based on the
\emph{Exponential Time Hypothesis} (ETH)~\cite{ImpagliazzoP01}. Roughly speaking, ETH is a conjecture that an $n$-variable instance of $3$-SAT cannot be solved by a deterministic $2^{o(n)}$-time algorithm (for a modern survey, see~\cite{LokshtanonMS13}). In our reductions, it will be convenient for us to work with the $k$-clique problem instead.
In the $k$-clique problem we are given a graph $G = (V,E)$ as an
input and the task is to decide whether there is a set of $k$ pairwise adjacent vertices in $V$. 
The naive algorithm for $k$-clique runs in $\Oh(n^k)$ time.
Even though the exact constant in the dependence on $k$ can be
improved~\cite{NesetrilP85}, ETH implies that the exponent must have a linear dependence on $k$.

\begin{theorem}[{\cite[Theorem 4.2]{ChenCFHJKX05}},~{\cite[Theorem 4.5]{ChenHKX06}}, and~{\cite[Theorem 14.21]{CyganFKLMPPS15}}]
    \label{thm:k-clique}
    Assuming the Exponential Time Hypothesis, there is no algorithm running in time $f(k)\cdot n^{o(k)}$ for the $k$-clique problem for any computable function $f$. 
    Moreover one can assume that $G$ is $k$-partite, \ie $G = (V_1\cup\ldots\cup V_k, E)$ and edges belong to $V_i \times V_j$ for $i \neq j \in \{1,\ldots,k\}$.
\end{theorem}

We will use Theorem~\ref{thm:k-clique} to show the following conditional lower bound for coverability in unary $d$-VASS, which is proved at the end of this section.
\begin{theorem}
	\label{thm:unary-kvass-coverability-lb}
	Assuming the Exponential Time Hypothesis, there does not exist an \mbox{$n^{2^{o(d)}}$-time} algorithm deciding coverability in a unary $d$-VASS with $n$ states.
\end{theorem}

We first reduce the $k$-clique problem to coverability in bounded 2-VASS with the ability to perform a fixed number of zero-tests.
We will then leverage a result by Rosier and Yen to construct an equivalent, with respect to coverability, $(\Oh(\log k))$-VASS without zero-tests. 

\begin{lemma}
	Given a $k$-partite graph $G = (V_1\cup \cdots \cup V_k, E)$ with $n$
    vertices, there exists a unary $(\Oh(n^{2k}), 2)$-VASS with $\Oh(k^2)$ zero-tests
    $\mathcal{T}$ such that there is a $k$-clique in $G$ if and only if there
    exists a run from $q_I(\vec{0})$ to $q_F(\vec{v})$ in $\mathcal{T}$, for
    some $\vec{v} \geq \vec{0}$.
	Moreover, $\norm{\mathcal{T}} \le \poly{n+k}$ and $\mathcal{T}$ can be constructed in $\poly{n+k}$-time.
	\label{lem:k-clique-reduction}
\end{lemma}
\begin{proof}
	Without loss of generality, we may assume that each of the $k$ vertex subsets in the graph has the same size $\abs{V_1} = \ldots = \abs{V_k} = \ell$. Thus $n = k \cdot \ell$.
    For convenience, we denote $V = \{1,\ldots,k\}\times\{1,\ldots,\ell\}$.

	We begin by sketching the main ideas behind the reduction before they are implemented.
    We start by finding the first $n = k \cdot \ell$ primes and associating a distinct prime $p_{i,j}$ to each vertex $(i,j) \in V$.
	Note that a product of $k$ different primes uniquely corresponds to selecting $k$ vertices.
	Thus the idea is to guess such a product, and test whether the corresponding verticies form a $k$-clique.
	To simplify the presentation we present VASS also as counter programs, inspired by Esparza's presentation of Lipton's lower bound~\cite[Section 7]{Esparza96}.

	We present an overview of our construction in Algorithm~\ref{algo_vass}. 
	Note that the counter $\vr{y}$ is used only by subprocedures. 
	Initially both counter values are $0$, as in the initial configuration of the coverability instance. 
	The program is non-deterministic and we are interested in the existence of a certain run. 
	One should think that coverability holds if and only if there is a run through the code without getting stuck so to say. 
	In this example a run can be stuck only in the \Edge{e} subprocedure, that will be explained later. 
	The precise final counter values are not important, as we are simply aiming to cover the target counter values $\vec{0}$.
	The variable $i$ (in the first loop) and variables $i$ and $j$ (in the second loop) are just syntactic sugar for copying similar code multiple times. 
	The variables $j$ (in the first loop) and $e$ (in the second loop) allow us to neatly represent non-determinism in a VASS.
	\begin{algorithm}[ht!]
		\SetKwInOut{Input}{input}
		\Input{$\vr{x} = 0$,  $\vr{y} = 0$}
		\BlankLine
		\increase{\vr{x}}{1}\;
		\For{$i \leftarrow 1$ \KwTo $k$}{
		    $\textbf{guess } j \in \{1,\ldots,\ell\}$ \;
		\Multiply{\vr{x}}{p_{i,j}}\;
		}
		\For{$(i,j) \in \{1,\ldots,k\}^2$, $i \neq j$}{
		$\textbf{guess } e \in E \cap (V_i \times V_j)$ \;
		\Edge{e}\;
		}
		\BlankLine
		\caption{A counter program for a VASS with zero tests with two counters $\vr{x}$ and $\vr{y}$.
		}
		\label{algo_vass}
	\end{algorithm}

	Algorithm~\ref{algo_vass} uses the \Multiply{\vr{x}}{p} and \Edge{e} subprocedures.
	These two subprocedures will be implemented later. Note that \Multiply{\vr{x}}{p} takes a counter $\vr{x}$ as input as we later reuse this subprocedure when there is more than one counter subject to multiplication.
	The intended behaviour of \Multiply{\vr{x}}{p} is that it can be performed if and only if as a result we get $\vr{x} = \vr{x} \cdot p$,
 	despite the fact that VASS can only additively increase and decrease counters. 
	The subprocedure \Edge{e} can be performed if and only if both vertices of the edge $e$ are encoded in the value of the counter $\vr{x}$.
	Overall, Algorithm~\ref{algo_vass} is designed so that in the first part the variable $\vr{x}$ is multiplied by $p_{i,j}$, where for every $i$ one $j$ is guessed. This equates to selecting one vertex from each $V_i$. 
	Then the second part the algorithm checks whether between every pair of selected vertices from $V_i$ and $V_j$ there is an edge. 
	Clearly there is a run through the program that does not get stuck if and only if there is $k$-clique in $G$. 

	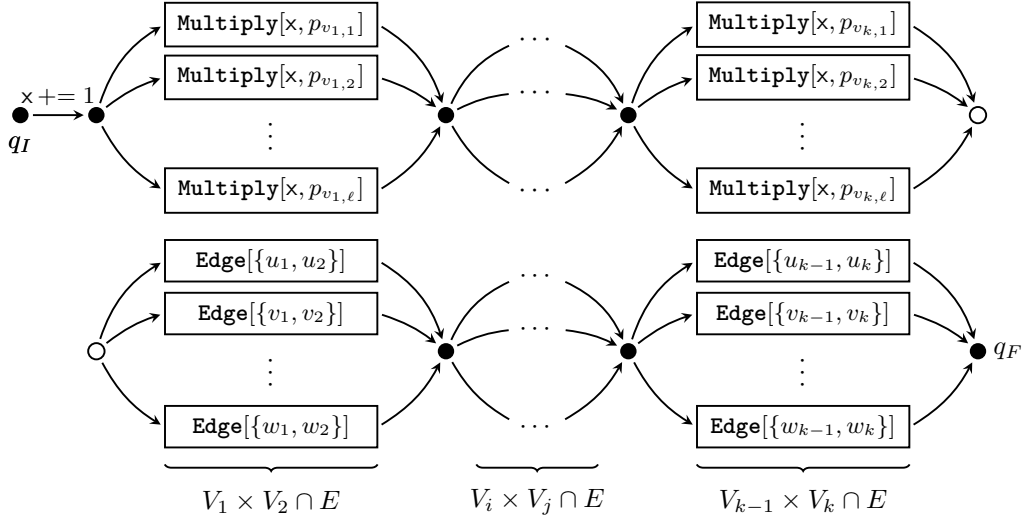
\begin{figure}[ht!]
		\centering
		\scalebox{1}{\begin{tikzpicture}
	\node[fill, circle, inner sep = 0.03in] (p) at (-1, 0) {};
	\node at (-1, -0.4) {$q_I$};
	\node[fill, circle, inner sep = 0.03in] (q1) at (0, 0) {};
	\node[fill, circle, inner sep = 0.03in] (q2) at (4.6, 0) {};
	\node[fill, circle, inner sep = 0.03in] (q3) at (7, 0) {};
	\node[draw, line width = 0.01in, circle, inner sep = 0.03in] (q4) at (11.6, 0) {};
	\node at (-1, -0.4) {$q_I$};
	\node at (12, 0) {\textcolor{white}{$q_F$}};

	\node[draw, rectangle, line width = 0.01in, minimum width = 1.1in] (g11) at (2.3, 1.2) {\small\Multiply{\vr{x}}{p_{v_{1,1}}}};
	\node[draw, rectangle, line width = 0.01in, minimum width = 1.1in] (g12) at (2.3, 0.5) {\small\Multiply{\vr{x}}{p_{v_{1,2}}}};
	\node[rotate = 90] (g1d) at (2.3, -0.25) {$\cdots$};
	\node[draw, rectangle, line width = 0.01in, minimum width = 1.1in] (g1n) at (2.3, -1) {\small\Multiply{\vr{x}}{p_{v_{1,\ell}}}};

	\node (g21) at (5.6, 1) {};
	\node (g22) at (5.6, 0.3) {};
	\node (g2n) at (5.6, -1) {};
	\node (qd) at (5.8, 1) {\large$\cdots$};
	\node (qd) at (5.8, 0.3) {\large$\cdots$};
	\node (qd) at (5.8, -1) {\large$\cdots$};
	\node (g31) at (6, 1) {};
	\node (g32) at (6, 0.3) {};
	\node (g3n) at (6, -1) {};

	\node[draw, rectangle, line width = 0.01in, minimum width = 1.1in] (gk1) at (9.3, 1.2) {\small\Multiply{\vr{x}}{p_{v_{k,1}}}};
	\node[draw, rectangle, line width = 0.01in, minimum width = 1.1in] (gk2) at (9.3, 0.5) {\small\Multiply{\vr{x}}{p_{v_{k,2}}}};
	\node[rotate = 90] (gkd) at (9.3, -0.25) {$\cdots$};
	\node[draw, rectangle, line width = 0.01in, minimum width = 1.1in] (gkn) at (9.3, -1) {\small\Multiply{\vr{x}}{p_{v_{k,\ell}}}};

	\draw[-stealth, line width = 0.01in, shorten <= 0.02in, shorten >= 0.02in] (p) -- node[above]{\small\increase{\vr{x}}{1}} (q1);

	\path[-stealth, line width = 0.01in, shorten <= 0.02in, shorten >= 0.02in] (q1) edge[bend left = 20] (g11.west);
	\path[-stealth, line width = 0.01in, shorten <= 0.02in, shorten >= 0.02in] (g11.east) edge[bend left = 20] (q2);
	\path[-stealth, line width = 0.01in, shorten <= 0.02in, shorten >= 0.02in] (q1) edge[bend left = 10] (g12.west);
	\path[-stealth, line width = 0.01in, shorten <= 0.02in, shorten >= 0.02in] (g12.east) edge[bend left = 10] (q2);
	\path[-stealth, line width = 0.01in, shorten <= 0.02in, shorten >= 0.02in] (q1) edge[bend right = 15] (g1n.west);
	\path[-stealth, line width = 0.01in, shorten <= 0.02in, shorten >= 0.02in] (g1n.east) edge[bend right = 15] (q2);

	\path[line width = 0.01in, shorten <= 0.02in, shorten >= 0.02in] (q2) edge[bend left = 20] (g21);
	\path[line width = 0.01in, shorten <= 0.02in, shorten >= 0.02in] (q2) edge[bend left = 10] (g22);
	\path[line width = 0.01in, shorten <= 0.02in, shorten >= 0.02in] (q2) edge[bend right = 15] (g2n);
	\path[-stealth, line width = 0.01in, shorten <= 0.02in, shorten >= 0.02in] (g31) edge[bend left = 20] (q3);
	\path[-stealth, line width = 0.01in, shorten <= 0.02in, shorten >= 0.02in] (g32) edge[bend left = 10] (q3);
	\path[-stealth, line width = 0.01in, shorten <= 0.02in, shorten >= 0.02in] (g3n) edge[bend right = 15] (q3);

	\path[-stealth, line width = 0.01in, shorten <= 0.02in, shorten >= 0.02in] (q3) edge[bend left = 20] (gk1.west);
	\path[-stealth, line width = 0.01in, shorten <= 0.02in, shorten >= 0.02in] (gk1.east) edge[bend left = 20] (q4);
	\path[-stealth, line width = 0.01in, shorten <= 0.02in, shorten >= 0.02in] (q3) edge[bend left = 10] (gk2.west);
	\path[-stealth, line width = 0.01in, shorten <= 0.02in, shorten >= 0.02in] (gk2.east) edge[bend left = 10] (q4);
	\path[-stealth, line width = 0.01in, shorten <= 0.02in, shorten >= 0.02in] (q3) edge[bend right = 15] (gkn.west);
	\path[-stealth, line width = 0.01in, shorten <= 0.02in, shorten >= 0.02in] (gkn.east) edge[bend right = 15] (q4);
\end{tikzpicture}}
		\\[3mm]
		\scalebox{1}{\begin{tikzpicture}
	\node[draw, line width = 0.01in, circle, inner sep = 0.03in] (q1) at (0, 0) {};
	\node[fill, circle, inner sep = 0.03in] (q2) at (4.6, 0) {};
	\node[fill, circle, inner sep = 0.03in] (q3) at (7, 0) {};
	\node[fill, circle, inner sep = 0.03in] (q4) at (11.6, 0) {};
		\node at (12, 0) {$q_F$};
		\node at (-1, -0.4) {\textcolor{white}{$q_I$}};

	\node[draw, rectangle, line width = 0.01in, minimum width = 1.1in] (g11) at (2.3, 1.2) {\small\Edge{\set{u_1, u_2}}};
	\node[draw, rectangle, line width = 0.01in, minimum width = 1.1in] (g12) at (2.3, 0.5) {\small\Edge{\set{v_1, v_2}}};
	\node[rotate = 90] (g1d) at (2.3, -0.25) {$\cdots$};
	\node[draw, rectangle, line width = 0.01in, minimum width = 1.1in] (g1n) at (2.3, -1) {\small\Edge{\set{w_1, w_2}}};
	\draw[line width = 0.01in, decoration={brace, mirror}, decorate] (0.9, -1.5) -- (3.7, -1.5);
	\node at (2.3, -2) {$V_1 \times V_2 \cap E$};

	\node (g21) at (5.6, 1) {};
	\node (g22) at (5.6, 0.3) {};
	\node (g2n) at (5.6, -1) {};
	\node (qd) at (5.8, 1) {$\cdots$};
	\node (qd) at (5.8, 0.3) {$\cdots$};
	\node (qd) at (5.8, -1) {$\cdots$};
	\node (g31) at (6, 1) {};
	\node (g32) at (6, 0.3) {};
	\node (g3n) at (6, -1) {};
	\draw[line width = 0.01in, decoration={brace, mirror}, decorate] (5, -1.5) -- (6.6, -1.5);
	\node at (5.8, -2) {$V_i \times V_j \cap E$};

	\node[draw, rectangle, line width = 0.01in, minimum width = 1.1in] (gk1) at (9.3, 1.2) {\small\Edge{\set{u_{k-1}, u_k}}};
	\node[draw, rectangle, line width = 0.01in, minimum width = 1.1in] (gk2) at (9.3, 0.5) {\small\Edge{\set{v_{k-1}, v_k}}};
	\node[rotate = 90] (gkd) at (9.3, -0.25) {$\cdots$};
	\node[draw, rectangle, line width = 0.01in, minimum width = 1.1in] (gkn) at (9.3, -1) {\small\Edge{\set{w_{k-1}, w_k}}};
	\draw[line width = 0.01in, decoration={brace, mirror}, decorate] (7.9, -1.5) -- (10.7, -1.5);
	\node at (9.3, -2) {$V_{k-1} \times V_{k} \cap E$};

	\path[-stealth, line width = 0.01in, shorten <= 0.02in, shorten >= 0.02in] (q1) edge[bend left = 20] (g11.west);
	\path[-stealth, line width = 0.01in, shorten <= 0.02in, shorten >= 0.02in] (g11.east) edge[bend left = 20] (q2);
	\path[-stealth, line width = 0.01in, shorten <= 0.02in, shorten >= 0.02in] (q1) edge[bend left = 10] (g12.west);
	\path[-stealth, line width = 0.01in, shorten <= 0.02in, shorten >= 0.02in] (g12.east) edge[bend left = 10] (q2);
	\path[-stealth, line width = 0.01in, shorten <= 0.02in, shorten >= 0.02in] (q1) edge[bend right = 15] (g1n.west);
	\path[-stealth, line width = 0.01in, shorten <= 0.02in, shorten >= 0.02in] (g1n.east) edge[bend right = 15] (q2);

	\path[line width = 0.01in, shorten <= 0.02in, shorten >= 0.02in] (q2) edge[bend left = 20] (g21);
	\path[line width = 0.01in, shorten <= 0.02in, shorten >= 0.02in] (q2) edge[bend left = 10] (g22);
	\path[line width = 0.01in, shorten <= 0.02in, shorten >= 0.02in] (q2) edge[bend right = 15] (g2n);
	\path[-stealth, line width = 0.01in, shorten <= 0.02in, shorten >= 0.02in] (g31) edge[bend left = 20] (q3);
	\path[-stealth, line width = 0.01in, shorten <= 0.02in, shorten >= 0.02in] (g32) edge[bend left = 10] (q3);
	\path[-stealth, line width = 0.01in, shorten <= 0.02in, shorten >= 0.02in] (g3n) edge[bend right = 15] (q3);

	\path[-stealth, line width = 0.01in, shorten <= 0.02in, shorten >= 0.02in] (q3) edge[bend left = 20] (gk1.west);
	\path[-stealth, line width = 0.01in, shorten <= 0.02in, shorten >= 0.02in] (gk1.east) edge[bend left = 20] (q4);
	\path[-stealth, line width = 0.01in, shorten <= 0.02in, shorten >= 0.02in] (q3) edge[bend left = 10] (gk2.west);
	\path[-stealth, line width = 0.01in, shorten <= 0.02in, shorten >= 0.02in] (gk2.east) edge[bend left = 10] (q4);
	\path[-stealth, line width = 0.01in, shorten <= 0.02in, shorten >= 0.02in] (q3) edge[bend right = 15] (gkn.west);
	\path[-stealth, line width = 0.01in, shorten <= 0.02in, shorten >= 0.02in] (gkn.east) edge[bend right = 15] (q4);
\end{tikzpicture}}
		\caption{The top part of the VASS implements the first line and the first loop in Algorithm~\ref{algo_vass}. 
		The variable $\vr{x}$ is multiplied by $k$ non-deterministically chosen primes $p_{i,j}$, each corresponding to a vertex in $V_i$. 
		The bottom part of the VASS implements the second loop in Algorithm~\ref{algo_vass}. 
		For every pair $i \neq j$ the VASS non-deterministically chooses $e \in V_i \cap V_j$ and invokes the subprocedure \Edge{e}.}
		\label{fig:clique-vass-implementation}
	\end{figure}
	In Figure~\ref{fig:clique-vass-implementation} we present a VASS with zero-tests implementing Algorithm~\ref{algo_vass}. 
	The construction will guarantee that $q_F(\vec{0})$ can be covered from $q_I(\vec{0})$ if and only if there is a $k$-clique in $G$.

	It remains to define the subprocedures. 
	One should think that every call of a subprocedure corresponds to a unique part of the VASS, like a gadget of sorts.
	To enter and leave the subprocedure one needs to add trivial transitions that to do not change the counter values.  
	All subprocedures rely on the invariant $\vr{y} = 0$ at the beginning and admit the invariant at the end.

	We start with \Multiply{\vr{x}}{p} and \Divide{\vr{x}}{p} that indeed multiply and divide $\vr{x}$ by $p$, respectively. 
	See Algorithm~\ref{alg:multiply-divide} for the counter program and VASS implementations.	
	Notice that the repeat loops correspond to the self-loops in the VASS. 
	In the \Multiply{\vr{x}}{p} gadget, it is easy to see that a run passes through the procedure if and only if the counter $\vr{x}$ is multiplied by $p$.
	Similarly, in the \Divide{\vr{x}}{p} gadget, it is easy to see that a run pass through the procedure if and only if the counter $\vr{x}$ is divided by $p$ wholly.
	Indeed, the division procedure would get stuck if $p \nmid \vr{x}$ because it will be impossible to exit the first loop.

	\pagebreak
	\begin{algorithm}[ht!]
		\begin{minipage}{.45\textwidth}
	  	\BlankLine
		\SetKwInOut{Input}{input}\SetKwInOut{Output}{output}
		\Input{$\vr{x} = v$, $\vr{y} = 0$}
		\Output{$\vr{x} = v \cdot p$,  $\vr{y} = 0$}
		\BlankLine
		\Repeat{$\vr{x}=0$}{
		\decrease{\vr{x}}{1}; \xspace\increase{\vr{y}}{1}\;
		}
		\Repeat{$\vr{y}=0$}{
		\increase{\vr{x}}{p}; \xspace\decrease{\vr{y}}{1}\;
		}
		\BlankLine
		\begin{tikzpicture}
	\node[fill, circle, inner sep = 0.03in] (am) at (0, 0) {};
	\node[fill, circle, inner sep = 0.03in] (bm) at (2, 0) {};
	\node[fill, circle, inner sep = 0.03in] (cm) at (4, 0) {};
	
	\node at (2, 2) {};
	\node at (2, -0.3) {};
	
	\draw[-stealth, line width = 0.01in, shorten <= 0.02in, shorten >= 0.02in] (am) -- node[above]{$\vr{x}=0$} (bm);
	\draw[-stealth, line width = 0.01in, shorten <= 0.02in, shorten >= 0.02in] (bm) -- node[above]{$\vr{y}=0$} (cm);

	\draw[-stealth, line width = 0.01in, shorten <= 0.02in, shorten >= 0.02in] (am) edge[loop above, distance = 0.5in, out=120, in=60]  (am);
	\node at (0, 1.7) {\decrease{\vr{x}}{1}};
	\node at (0, 1.2) {\increase{\vr{y}}{1}};
	\draw[-stealth, line width = 0.01in, shorten <= 0.02in, shorten >= 0.02in] (bm) edge[loop above, distance = 0.5in, out=120, in=60] (bm);
	\node at (2, 1.7) {\increase{\vr{x}}{p}};
	\node at (2, 1.2) {\decrease{\vr{y}}{1}};
\end{tikzpicture}
		\end{minipage}\begin{minipage}{.1\textwidth}
			\begin{tikzpicture}
				\draw[color=gray] (0, 0) -- (0, 6.25);
			\end{tikzpicture}
		\end{minipage}\begin{minipage}{.45\textwidth}
	  	\BlankLine
		\SetKwInOut{Input}{input}\SetKwInOut{Output}{output}
		\Input{$\vr{x} = v \cdot p$, $\vr{y} = 0$}
		\Output{$\vr{x} = v$,  $\vr{y} = 0$}
		\BlankLine
		\Repeat{$\vr{x}=0$}{
		\decrease{\vr{x}}{p}; \xspace\increase{\vr{y}}{1}\;
		}
		\Repeat{$\vr{y}=0$}{
		\increase{\vr{x}}{1}; \xspace\decrease{\vr{y}}{1}\;
		}
		\BlankLine
		\begin{tikzpicture}
	\node[fill, circle, inner sep = 0.03in] (ad) at (0, 0) {};
	\node[fill, circle, inner sep = 0.03in] (bd) at (2, 0) {};
	\node[fill, circle, inner sep = 0.03in] (cd) at (4, 0) {};
	
	\node at (2, 2) {};
	\node at (2, -0.3) {};
	
	\draw[-stealth, line width = 0.01in, shorten <= 0.02in, shorten >= 0.02in] (ad) -- node[above]{$\vr{x}=0$} (bd);
	\draw[-stealth, line width = 0.01in, shorten <= 0.02in, shorten >= 0.02in] (bd) -- node[above]{$\vr{y}=0$} (cd);

	\draw[-stealth, line width = 0.01in, shorten <= 0.02in, shorten >= 0.02in] (ad) edge[loop above, distance = 0.5in, out=120, in=60] (ad);
	\node at (0, 1.7) {\decrease{\vr{x}}{p}};
	\node at (0, 1.2) {\increase{\vr{y}}{1}};
	\draw[-stealth, line width = 0.01in, shorten <= 0.02in, shorten >= 0.02in] (bd) edge[loop above, distance = 0.5in, out=120, in=60] (bd);
	\node at (2, 1.7) {\increase{\vr{x}}{1}};
	\node at (2, 1.2) {\decrease{\vr{y}}{1}};
\end{tikzpicture}
		\end{minipage}
		\caption{The counter program of \Multiply{\vr{x}}{p} above its VASS implementation (left) and the counter program of \Divide{\vr{x}}{p} above its VASS implementation (right).}
		\label{alg:multiply-divide}
	\end{algorithm}

	The procedure \Edge{\set{u,v}} is very simple, it is a sequence of four subprocedures, see Algorithm~\ref{algo_edge}. 
	Indeed, to check if the vertices from edge $e$ are encoded in $\vr{x}$ we simply check whether $\vr{x}$ is divisible by the corresponding primes. 
	Afterwards we multiply $\vr{x}$ with the same primes so that the value does not change and it is ready for future edge checks.
	\begin{algorithm}[ht!]
		\begin{minipage}{.5\textwidth}
			\SetKwInOut{Input}{input}\SetKwInOut{Output}{output}
			\Input{$\vr{x} = v$, $\vr{y} = 0$}
			\Output{$\vr{x} = v$, $\vr{y} = 0$}
			\BlankLine
			\Divide{\vr{x}}{p_u} \;
			\BlankLine
			\Multiply{\vr{x}}{p_u} \;
			\BlankLine
			\Divide{\vr{x}}{p_v} \;
			\BlankLine
			\Multiply{\vr{x}}{p_v} \;
			\BlankLine
		\end{minipage}\begin{minipage}{.5\textwidth}
			\begin{tikzpicture}
	\node[fill, circle, inner sep = 0.03in] (p) at (-0.5, 0) {};
	\node[fill, circle, inner sep = 0.03in] (q) at (4.5, -3) {};
	\node at (0, -3.5) {};

	\node[draw, rectangle, line width = 0.01in, minimum width = 1.1in] (du) at (2, 0) {\small\Divide{\vr{x}}{p_u}};
	\node[draw, rectangle, line width = 0.01in, minimum width = 1.1in] (mu) at (2, -1) {\small\Multiply{\vr{x}}{p_u}};
	\node[draw, rectangle, line width = 0.01in, minimum width = 1.1in] (dv) at (2, -2) {\small\Divide{\vr{x}}{p_v}};
	\node[draw, rectangle, line width = 0.01in, minimum width = 1.1in] (mv) at (2, -3) {\small\Multiply{\vr{x}}{p_v}};

	\draw[-stealth, line width = 0.01in, shorten <= 0.02in, shorten >= 0.02in] (p) -- (du);
	\draw[-stealth, line width = 0.01in, shorten <= 0.02in, shorten >= 0.02in] (du) -- (mu);
	\draw[-stealth, line width = 0.01in, shorten <= 0.02in, shorten >= 0.02in] (mu) -- (dv);
	\draw[-stealth, line width = 0.01in, shorten <= 0.02in, shorten >= 0.02in] (dv) -- (mv);
	\draw[-stealth, line width = 0.01in, shorten <= 0.02in, shorten >= 0.02in] (mv) -- (q);
\end{tikzpicture}
		\end{minipage}
		\caption{The counter program for \Edge{\set{u, v}} and its VASS implementation.}
		\label{algo_edge}
	\end{algorithm}
	
	It remains to analyse the size of the VASS and its construction time in this reduction time.
    In every run from $q_I(\vec{0})$ to $q_F(\vec{v})$, for some $\vec{v} \geq \vec{0}$, the greatest counter value observable can be bounded above by $p^k$ where $p$ is the $n$-th prime.
    By the Prime Number Theorem (for example, see~\cite{Zagier79}), we know that $p^k \leq \Oh((n\log(n))^k) \leq \Oh(n^{2k})$ is an upper bound on the counter values observed.
    Hence $\mathcal{T}$ is an $\Oh(n^{2k})$-bounded unary 2-VASS.

    Now, we count the number of zero-tests performed in each run from $q_I(\vec{0})$ to $q_F(\vec{v})$, for some $\vec{v} \geq
    \vec{0}$.
    The only zero-tests occur in the instances of the \MultiplyGadget and \DivideGadget subprocedures, each performing two zero-tests.
    In the first part of $\mathcal{T}$, a run will encounter $k$ many \MultiplyGadget subprocedures, contributing $2k$ many zero-tests. 
    In the second part of $\mathcal{T}$, a run will encounter $k \choose 2$ many \EdgeGadget subprocedures, each containing two \MultiplyGadget subprocedures and two \DivideGadget subprocedures, in total contributing $8 {k \choose 2}$ many zero-tests.
    Together, every run encounters exactly $2k + 8 {k \choose 2} = 2k(2k-1)$ many zero-tests.
    Hence $\mathcal{T}$ is an $\Oh(n^{2k})$-bounded unary 2-VASS with $2k(2k-1)$ zero-tests.

	Finally, the \MultiplyGadget and \DivideGadget subprocedures contain three states and five transitions.
    Since the $n$-th prime is bounded above by $\Oh(n\log(n))$, we also get
    $\norm{\mathcal{T}} = \Oh(n\log(n))$, hence our VASS can be
    represented using unary encoding.
	Analysing Algorithm~\ref{algo_vass}, it is easy to see that overall the number of states is polynomial in $n$.
    Finally, the first $n$ primes can be found in $\Oh(n^{1+o(1)})$ time~\cite{AgrawalKS04}.
	Therefore, in total $\mathcal{T}$ has size $\norm{\mathcal{T}} = \poly{n+k}$ and can be constructed in $\poly{n+k}$ time.
\end{proof}

To attain conditional lower bounds for coverability we must replace the
zero-tests. 
We make use of a technique of Rosier and Yen~\cite{RosierY86} that relies on the construction of Lipton~\cite{Lipton76}.
They show that a $(2n)^{2^k}$-bounded counter machine with finite state control can be simulated by a unary $(\Oh(k))$-VASS with $n$ states. 
As Rosier and Yen detail after their proof, it is possible to apply this technique to multiple counters with zero-tests at once~\cite{RosierY86}. 
This accordingly results in the number of VASS counters increasing, but we instantiate this with just two counters.
We remark that the VASS constructed in Lemma~\ref{lem:k-clique-reduction} is structurally bounded, so for any initial configuration there is a limit on the largest observable counter value, as is the VASS Lipton constructed.

\begin{lemma}[Corollary of~{\cite[Lemma 4.3]{RosierY86}}] \label{lem:rosier-yen}
	Let $\mathcal T$ be an $n$-state unary $(n^{\Oh(k)}, 2)$-VASS with zero-tests, for some parameter $k$.
    Then there exists an $\Oh(n)$-state $(\Oh(\log{k}))$-VASS $\mathcal{V}$, such
    that there is a run from $q_I(\vec{0})$ to $q_F(\vec{v})$, for some $\vec{v}
    \geq \vec{0}$, in $\mathcal{T}$ if and only if there is a run from $q_I'(\vec{0})$ to $q_F'(\vec{w})$, for some $\vec{w} \geq \vec{0}$, in $\mathcal{V}$.
    Moreover, $\mathcal{V}$ has size $\Oh(\abs{\mathcal{T}})$ and can be constructed in the same time.
\end{lemma}

With this, we can finish the proof of our main theorem for this section.

\begin{proof}[Proof of Theorem~\ref{thm:unary-kvass-coverability-lb}]
	Let $k = 2^d$. 
	We instantiate Lemma~\ref{lem:k-clique-reduction} on $k$-partite graphs $G$ with $n$ vertices.
	We therefore obtain a unary $(n^{2^{\Oh(d)}}, 2)$-VASS with zero tests $\mathcal{T}$ such that
    $G$ contains a $k$-clique if and only if there is a run from $q_I(\vec{0})$ to $q_F(\vec{v})$, for some $\vec{v} \geq \vec{0}$, in $\mathcal{T}$.

	Given the bound on the value of the counters, we can apply Lemma~\ref{lem:rosier-yen} to $\mathcal{T}$.
	This gives us an $\Oh(n)$-state $(\Oh(d))$-VASS $\mathcal{V}$ such that $G$ contains a $k$-clique if and only if there is a run from $q_I'(\vec{0})$ to $q_F'(\vec{w})$, for some $\vec{w} \geq \vec{0}$, in $\mathcal{V}$.

    By Theorem~\ref{thm:k-clique} we conclude that under the Exponential Time Hypothesis there does not exist an $n^{2^{o(d)}}$-time algorithm deciding coverability in unary $d$-VASS.
\end{proof}

\section{Coverability and Reachability in Bounded Unary VASS}
\label{sec:bounded-vass}

In this section, we give even tighter bounds for coverability in \emph{bounded} fixed dimension unary VASS. 
Specifically, for a time constructible function $B(n)$, the coverability problem in $(B(n), d)$-VASS asks, for a given $(B(n), d)$-VASS $\Vv = (Q, T)$ of size $n$ as well as configurations $p(\vec{u})$, $q(\vec{v})$, whether there is a run in $\Vv$ from $p(\vec{u})$ to $q(\vec{v}')$ for some $\vec{v}' \geq \vec{v}$ such that each counter value remains in $\set{0,\ldots, B(n)}$ throughout. 
We would like to clarify the fact that the bound is not an input parameter.
We focus on the natural setting of linearly-bounded fixed dimension VASS, that is $(\Oh(n), d)$-VASS.
There is simple algorithm, presented in the proof of in Observation~\ref{obs:baseline_bounded}, that yields an immediate $\Oh(n^{d+1})$ upper bound for the time needed to decide the coverability problem.
We accompany this observation with closely matching lower bounds, please see Table~\ref{tab:coverability-time-bounds} for an overview.
\begin{table}[ht]
	\centering
	\begin{tabular}{c | c | c c}
		$d$ & Lower Bound & Upper Bound \\ \hline
		$0$ & $\Omega(n)$ (trivial) & $\Oh(n)$ \\
		$1$ & $n^{2-o(1)}$ (Theorem~\ref{thm:k-cycle-to-vass}) & $\Oh(n^2)$ \\
		$2$ & $n^{2-o(1)}$ (from above) & $\Oh(n^3)$ \\
		$3$ & $n^{2-o(1)}$ (from above) & $\Oh(n^4)$ \\
		$d \geq 4$ & $n^{d-2-o(1)}$ (Theorem~\ref{thm:bounded-reachability}) & $\Oh(n^{d+1})$
	\end{tabular}
	\caption{Conditional lower bounds and upper bounds of the time complexity of coverability and reachability in unary $(\Oh(n),d)$-VASS. 
	For clarity, we remark that Theorem~\ref{thm:k-cycle-to-vass} is subject to Hypothesis~\ref{hyp:k-cycle-hypothesis} and that Theorem~\ref{thm:bounded-reachability} is subject to Hypothesis~\ref{hyp:hyperclique}.
	Note that the lower bounds for dimensions $d=2$ and $d=3$ follow from Theorem~\ref{thm:k-cycle-to-vass} by just adding components consisting of only zeros.
	All upper bounds follow from Observation~\ref{obs:baseline_bounded}.}
	\label{tab:coverability-time-bounds}
\end{table}

\begin{observation}\label{obs:baseline_bounded}
Coverability in an $n$-sized unary $(B(n), d)$-VASS can be solved in $\Oh(n(B(n)+1)^d)$-time.
\end{observation}
\begin{proof}
    Since all configuration in a $(B(n), d)$-VASS belong to the finite set $Q \times \{0,\ldots,B(n)\}^d$, we can exhaustively explore all configurations reachable from $p(\vec{v})$ using a straightforward depth-first search. 
	Each state $q \in Q$ and each transition $t \in T$ will be considered at most once for each admissible vector in $\{0,\ldots,B(n)\}^d$, requiring time $\Oh(n(B(n)+1)^d)$ since $\abs{Q}, \abs{T} \leq n$. 
	We accept the instance if and only if we ever witnessed a configuration $q(\vec{v}')$ for some $\vec{v}' \geq \vec{v}$.  
\end{proof}

\subsection*{Lower Bounds for Coverability in Linearly-Bounded VASS}

Now, we consider lower bounds for the coverability problem in linearly-bounded fixed dimension unary VASS. 
Firstly, in dimension one, we show that quadratic running time is conditionally optimal under the $k$-cycle hypothesis. 
Secondly, in dimensions four and higher, we require a running time at least $n^{d-2-o(1)}$ under the 3-uniform hyperclique hypothesis. 
Together, this provides evidence that the simple $\Oh(n^{d+1})$ algorithm for coverability in $(\Oh(n), d)$-VASS is close to optimal, as summarised in Table~\ref{tab:coverability-time-bounds}.

\begin{hypothesis}[$k$-Cycle Hypothesis]\label{hyp:k-cycle-hypothesis}
	For every $\varepsilon > 0$, there exists some $k$ such that there does not exist a $\Oh(m^{2-\varepsilon})$-time algorithm for finding a $k$-cycle in directed graphs with $m$ edges.
\end{hypothesis}

The $k$-cycle hypothesis arises from the state-of-the-art $\Oh(m^{2-\frac{c}{k}+o(1)})$-time algorithms, where $c$ is some constant~\cite{AlonYZ97,YusterZ04,DalirrooyfardVW21}. 
It has been previously used as an assumption for hardness results, for example, see~\cite{LincolnVWW18,AnconaHRWW19,DalirrooyfardJVWW22}.
It is a standard observation, due to colour-coding arguments, that we may without loss of generality assume that the graph given is a $k$-circle-layered graph~{\cite[Lemma 2.2]{LincolnVWW18}}.
Specifically, we can assume that the input graph $G=(V, E)$ has vertex partition $V=V_0\cup \cdots \cup V_{k-1}$ such that each edge $\set{u, v} \in E$ is in $V_i \times V_{i+1\Mod{k}}$ for some $0 \le i < k$. 
Furthermore, we may assume $|V| \leq |E|$.

\begin{lemma}\label{lem:k-cycle-reduction}
	Given a $k$-circle-layered graph $G = (V_0 \cup \cdots \cup V_{k-1}, E)$
	with $m$ edges, there exists a unary $(\Oh(n), 1)$-VASS $\Vv$ such that there is a $k$-cycle in $G$ if and only if there exists a run from $p(0)$ to $q(0)$ in $\Vv$. 
    Moreover, $\Vv$ has size $n \leq \Oh(m)$ and can be constructed in $\Oh(m)$ time.
\end{lemma}
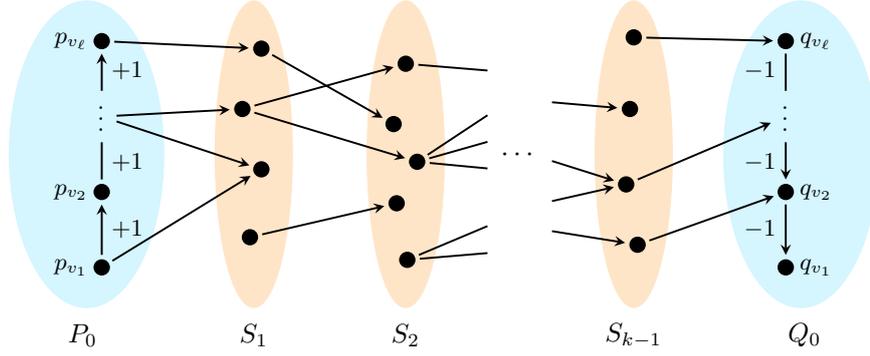
\begin{figure}[ht!]
	\centering
	\begin{tikzpicture}
	\draw[fill, lightblue, opacity = 0.5] (-0.2, 0) ellipse (0.4in and 0.8in);
	\node at (-0.25, -2.4) {$P_0$};
	\node[fill, circle, inner sep = 0.03in] (pv01) at (0, -1.5) {};
	\node at (-0.4, -1.5) {\small $p_{v_1}$};
	\node[fill, circle, inner sep = 0.03in] (pv02) at (0, -0.5) {};
	\node at (-0.4, -0.5) {\small $p_{v_2}$};
	\node[rotate = 90] (pv0d) at (0, 0.5) {\small $\cdots$};
	\node[fill, circle, inner sep = 0.03in] (pv0n) at (0, 1.5) {};
	\node at (-0.4, 1.5) {\small $p_{v_\ell}$};

	\draw[-stealth, line width = 0.01in, shorten >= 0.02in, shorten <= 0.02in] (pv01) -- node[right]{\small $+1$} (pv02);
	\draw[line width = 0.01in, shorten <= 0.02in] (pv02) -- node[right]{\small $+1$} (pv0d);
	\draw[-stealth, line width = 0.01in, shorten >= 0.02in] (pv0d) -- node[right]{\small $+1$} (pv0n);

	\draw[fill, lightblue, opacity = 0.5] (9.2, 0) ellipse (0.4in and 0.8in);
	\node at (9.25, -2.4) {$Q_0$};
	\node[fill, circle, inner sep = 0.03in] (qv0n) at (9, 1.5) {};
	\node at (9.4, 1.5) {\small $q_{v_\ell}$};
	\node[rotate=90] (qv0d) at (9, 0.5) {\small $\cdots$};
	\node[fill, circle, inner sep = 0.03in] (qv02) at (9, -0.5) {};
	\node at (9.4, -0.5) {\small $q_{v_2}$};
	\node[fill, circle, inner sep = 0.03in] (qv01) at (9, -1.5) {};
	\node at (9.4, -1.5) {\small $q_{v_1}$};
	\draw[line width = 0.01in, shorten <= 0.02in] (qv0n) -- node[left]{\small $-1$} (qv0d);
	\draw[-stealth, line width = 0.01in, shorten >= 0.02in] (qv0d) -- node[left]{\small $-1$} (qv02);
	\draw[-stealth, line width = 0.01in, shorten >= 0.02in, shorten <= 0.02in] (qv02) -- node[left]{\small $-1$} (qv01);

	\draw[fill, lightorange, opacity = 0.5] (2,0) ellipse (0.2in and 0.8in);
	\node at (2, -2.4) {$S_1$};
	\node[fill, circle, inner sep = 0.03in] (sv11) at (2.1, 1.4) {};
	\node[fill, circle, inner sep = 0.03in] (sv12) at (1.85, 0.6) {};
	\node[fill, circle, inner sep = 0.03in] (sv13) at (2.1, -0.2) {};
	\node[fill, circle, inner sep = 0.03in] (sv14) at (1.95, -1.1) {};

	\draw[fill, lightorange, opacity = 0.5] (4, 0) ellipse (0.2in and 0.8in);
	\node at (4, -2.4) {$S_2$};
	\node[fill, circle, inner sep = 0.03in] (sv21) at (4, 1.2) {};
	\node[fill, circle, inner sep = 0.03in] (sv22) at (3.84, 0.4) {};
	\node[fill, circle, inner sep = 0.03in] (sv23) at (4.15, -0.1) {};
	\node[fill, circle, inner sep = 0.03in] (sv24) at (3.88, -0.65) {};
	\node[fill, circle, inner sep = 0.03in] (sv25) at (4.02, -1.4) {};
	\node (sd21) at (5.2, 1.1) {};
	\node (sd22) at (5.2, 0.6) {};
	\node (sd23) at (5.2, 0.2) {};
	\node (sd24) at (5.2, -0.2) {};
	\node (sd25) at (5.2, -0.9) {};
	\node (sd26) at (5.2, -1.3) {};

	\node at (5.5, 0) {\large $\cdots$};

	\node (sd31) at (5.8, 1.3) {};
	\node (sd32) at (5.8, 0.7) {};
	\node (sd33) at (5.8, -0.05) {};
	\node (sd34) at (5.8, -0.65) {};
	\node (sd35) at (5.8, -1) {};
	\draw[fill, lightorange, opacity = 0.5] (7,0) ellipse (0.2in and 0.8in);
	\node at (7, -2.4) {$S_{k-1}$};
	\node[fill, circle, inner sep = 0.03in] (svk1) at (7, 1.55) {};
	\node[fill, circle, inner sep = 0.03in] (svk2) at (6.95, 0.6) {};
	\node[fill, circle, inner sep = 0.03in] (svk3) at (6.9, -0.4) {};
	\node[fill, circle, inner sep = 0.03in] (svk4) at (7.05, -1.2) {};

	\draw[-stealth, line width = 0.01in, shorten >=0.02in, shorten <= 0.02in] (pv01) -- (sv13);
	\draw[-stealth, line width = 0.01in, shorten >= 0.02in] (pv0d) -- (sv12);
	\draw[-stealth, line width = 0.01in, shorten >= 0.02in] (pv0d) -- (sv13);
	\draw[-stealth, line width = 0.01in, shorten >=0.02in, shorten <= 0.02in] (pv0n) -- (sv11);
	\draw[-stealth, line width = 0.01in, shorten >=0.02in, shorten <= 0.02in] (svk1) -- (qv0n);
	\draw[-stealth, line width = 0.01in, shorten <= 0.02in] (svk3) -- (qv0d);
	\draw[-stealth, line width = 0.01in, shorten >=0.02in, shorten <= 0.02in] (svk4) -- (qv02);

	\draw[-stealth, line width = 0.01in, shorten >= 0.02in, shorten <= 0.02in] (sv11) -- (sv22);
	\draw[-stealth, line width = 0.01in, shorten >= 0.02in, shorten <= 0.02in] (sv12) -- (sv21);
	\draw[-stealth, line width = 0.01in, shorten >= 0.02in, shorten <= 0.02in] (sv12) -- (sv23);
	\draw[-stealth, line width = 0.01in, shorten >= 0.02in, shorten <= 0.02in] (sv14) -- (sv24);
	\draw[line width = 0.01in, shorten <= 0.02in] (sv21) -- (sd21);
	\draw[line width = 0.01in, shorten <= 0.02in] (sv23) -- (sd22);
	\draw[line width = 0.01in, shorten <= 0.02in] (sv23) -- (sd23);
	\draw[line width = 0.01in, shorten <= 0.02in] (sv23) -- (sd24);
	\draw[line width = 0.01in, shorten <= 0.02in] (sv25) -- (sd25);
	\draw[line width = 0.01in, shorten <= 0.02in] (sv25) -- (sd26);
	\draw[-stealth, line width = 0.01in, shorten >= 0.02in] (sd32) -- (svk2);
	\draw[-stealth, line width = 0.01in, shorten >= 0.02in] (sd33) -- (svk3);
	\draw[-stealth, line width = 0.01in, shorten >= 0.02in] (sd34) -- (svk3);
	\draw[-stealth, line width = 0.01in, shorten >= 0.02in] (sd35) -- (svk4);
\end{tikzpicture}
	\caption{The $(\Oh(n), 1)$-VASS $\Vv$ of size $n \le \Oh(m)$ for finding $k$-cycle in a $k$-circle-layered graphs with $m$ edges.
	Note that unlabelled transitions have zero effect.
	Observe that the graph is mostly copied into the states and transitions of the linearly-bounded 1-VASS.  
	Importantly, two copies of $V_0$ are created.
	By starting at $p_{v_1}(0)$ in the first copy, a vertex from $V_0$ belonging to the $k$-cycle can be selected by loading the sole counter with a value corresponding to that vertex.
	Then, in the second copy, $q_{v_1}(0)$ can only be reached if the state first arrived at corresponds to the vertex selected in the beginning.
	Accordingly, there is a run from $p_{v_1}(0)$ to $q_{v_1}(0)$ if and only if there exists a $k$-cycle, since the states visited in the underlying path of the run correspond to the vertices of the $k$-cycle.}
	\label{fig:1vass}
\end{figure}
\begin{proof}
	Consider the unary $(\Oh(m), 1)$-VASS $\Vv = (Q, T)$ that is defined as follows, please also refer back to Figure~\ref{fig:1vass}.
	For ease of construction let us number the vertices in $V_0$, so suppose that $V_0 = \set{v_1, \ldots, v_\ell}$.

	Let us define the set of states $Q$.
	There are two copies of the vertex subset $V_0$, namely $P_0 = \set{p_{v_1}, \ldots, p_{v_\ell}}$ and $Q_0 = \set{q_{v_1}, \ldots, q_{v_\ell}}$.
	There are also copies of each of the vertex subsets $V_1, V_2, \ldots, V_{k-1}$, namely $S_i = \set{s_v : v \in V_i}$ for each $1 \leq i \leq k-1$.
	\begin{equation*}
		Q = P_0 \cup S_1 \cup S_2 \cup \cdots \cup S_{k-1} \cup Q_0
	\end{equation*}

	Now, we define the set of transitions $T$.
    There are three kinds of transitions, the initial \emph{vertex selection}
    transitions $T_I$, the intermediate transitions $T_E$, and the final
    \emph{vertex checking} transitions $T_F$.
	\begin{equation*}
		T = T_I \cup T_E \cup T_F
	\end{equation*}
    The initial transitions connect states in $P_0$ sequentially. Each
    transition increments the counter.
	Intuitively speaking, the counter takes a value corresponding to the vertex in $V_0$ that will belong to the $k$-cycle in $G$.
	\begin{equation*}
		T_I = \set{(p_{v_i}, 1, p_{v_{i+1}}) : 1 \leq i < \ell}
	\end{equation*}
	The intermediate transitions are directed copies of the edges in the original graph. 
	The only difference is that edges between $V_0$ and $V_1$ are now
    transitions from $P_0$ to $S_1$ and edges between $V_{k-1}$ and $V_0$ become transitions from $S_{k-1}$ to $Q_0$.
	\begin{align*}
		T_E = \, 
		& \set{ (p_u, 0, s_v) : \set{u, v} \in V_0 \times V_1 } \, \cup \, \set{ (s_u, 0, q_v) : \set{u, v} \in V_{k-1} \times V_0 } \, \cup \\
		& \set{ (s_u, 0, s_v) : \set{u, v} \in V_i \times V_{i+1} \text{ for some } 1 \leq i < k-1 } 
	\end{align*}
    The final transitions connect the states in $Q_0$ sequentially. Each such
    transition decrements the counter.
	Intuitively speaking, if the state reached in $Q_0$ matches the counter that has a value corresponding to the vertex in $V_0$ then the final state $q_{v_1}$ can be reached with counter value zero.
	\begin{equation*}
		T_F = \set{(q_{v_{i+1}}, -1, q_{v_i}) : 1 \leq i < \ell}
	\end{equation*}

	Importantly, there is a run from the initial configuration $p_{v_1}(0)$ to the target configuration $q_{v_1}(0)$ in $\Vv$ if and only if there is a $k$-cycle in the $k$-circle-layered graph $G$.
	In closing, observe that $\abs{Q} \leq 2|V|$ and $\abs{T} \leq 2|V| + |E|$.
	Therefore, $\Vv$ has size $\Oh(m)$.
	We remark that the greatest possible counter value is trivially bounded
    above by $\abs{Q}$, hence $\Vv$ is a unary $(\Oh(m),1)$-VASS of size $\Oh(m)$.
\end{proof}

\begin{theorem} \label{thm:k-cycle-to-vass}
	Assuming the $k$-cycle hypothesis, there does not exist an
    $\Oh(n^{2-o(1)})$-time algorithm deciding coverability or reachability in
    unary $(\Oh(n), 1)$-VASS of size $n$.
\end{theorem}
\begin{proof}
	Assume for contradiction that reachability in a unary $(\Oh(n),1)$-VASS of size $n$ can be solved in time $\Oh(n^{2-\varepsilon})$ for some $\varepsilon > 0$. 
	By the $k$-cycle hypothesis (Hypothesis~\ref{hyp:k-cycle-hypothesis}), there exists a $k$ such that the problem of finding a $k$-cycle in a $k$-circle layered graph with $m$ vertices cannot be solved in time $\Oh(m^{2-\varepsilon})$. 
	Via the reduction presented above in Lemma~\ref{lem:k-cycle-reduction}, we
    create a $(\Oh(n), 1)$-VASS $\Vv$ of size $n \leq \Oh(m)$ together with an initial configuration $p(0)$ and a target configuration $q(0)$, such that deciding reachability from $p(0)$ to $q(0)$ in $\Vv$ determines the existence of a $k$-cycle in $G$. 
	Thus the $\Oh(n^{2-\varepsilon})$ algorithm for reachability would give a
    $\Oh(m^{2-\varepsilon})$ algorithm for finding $k$-cycles, contradicting the $k$-cycle hypothesis.

	By the equivalence of coverability and reachability in unary $(\Oh(n), 1)$ VASS in Lemma~\ref{lem:equiv_bounded_reach_cover}, the same lower bound holds for coverability.
\end{proof}

\begin{corollary} \label{cor:unary-2vass-coverability}
	Assuming the $k$-cycle hypothesis, there does not exist an
    $\Oh(n^{2-o(1)})$-time algorithm for coverability in unary 2-VASS of size $n$.
\end{corollary}
\begin{proof}
	Consider a standard modification of the reduction presented for Lemma~\ref{lem:k-cycle-reduction}, that is to increase the dimension of $\Vv = (Q, T)$ by one by adding an opposite counter of sorts, yielding a 2-VASS $\Ww = (Q, T')$.
	For every transition $(p, t, q) \in T$, create a transition also modifying the opposite counter, $(p, (t, -t), q) \in T'$.
	Now, the instance $(\Ww, p(0, n), q(0, n))$ of coverability holds if and only if the instance $(\Vv, p(0), q(0))$ of reachability holds.
	The rest follows by Theorem~\ref{thm:k-cycle-to-vass}.
\end{proof}
 
Reachability in $(\Oh(n), d)$-VASS can be decided in $\Oh(n(B(n)+1)^d)$-time using the simple algorithm in Observation~\ref{obs:baseline_bounded} with a trivially modified acceptance condition.
It turns out that coverability and reachability are equivalent in unary $(\Oh(n), d)$-VASS.
This is true in the sense that it may hold that for example coverability in a $(100n,d)$-VASS may be reduced in linear time to reachability in a $(3n, d)$-VASS.
Conversely, reachability in some linearly-bounded $d$-VASS can be reduced in linear time to a corresponding instance of coverability in a linearly-bounded $d$-VASS.
Note that perversely, it appears plausible that instances of coverability in a $(100n,d)$-VASS could in fact be simpler to solve than in a $(3n, d)$-VASS.

\begin{lemma}\label{lem:equiv_bounded_reach_cover}
	For a $(B(n),d)$-VASS, let $C^{B(n)}(n)$ and $R^{B(n)}(n)$ denote the optimal running times for coverability and reachability, respectively.
	For any $\gamma > 0$, there exists some $\delta > 0$ such that $C^{\gamma\cdot
    n}(n) \leq \Oh(R^{\delta\cdot n}(n))$. 
	Conversely, for any $\gamma > 0$, there exists some $\delta>0$ such that
	$R^{\gamma\cdot n}(n) \le \Oh(C^{\delta\cdot n}(n))$.
\end{lemma}
\begin{proof}
	Given an instance $(\Vv, \config{p}{u}, \config{p}{v})$, of size $n$, of coverability in $B(n)$-bounded VASS $\Vv$.
	We construct a $B(n)$-bounded VASS $\Vv'$ from $\Vv$ by adding transitions $(q,-\vec{e}_i,q)$ for every $1 \le i \le d$. 
	It is easy to see that there exists a run from $p(\vec{u})$ to $q(\vec{v})$ in $\Vv'$ if and only if there exists and a run from $p(\vec{u})$ to $q(\vec{v}')$ for some $\vec{v}' \ge \vec{v}$, in $\Vv$. 
	Since $\norm{\Vv'} = \Oh(n)$, for $B(n) = \gamma\cdot n$ we can ensure that
	$B(n) = \delta\cdot\norm{\Vv'}$ for an appropriately selected $\delta$. 
	Thus, $C^{\gamma n}(n) \leq \Oh(R^{\delta n}(\Oh(n))) \leq \Oh(R^{\delta n}(n))$.

    Conversely, consider an instance $(\Vv, \config{p}{u}, \config{p}{v})$, of size $n$, of reachability in a $B(n)$-bounded VASS $\Vv$, again denote $n = \norm{\Vv}$.
    We construct the VASS $B(n)$-bounded VASS $\Vv'$ from $\Vv$ by adding a path from $q$ to a new state $r$ whose transitions update the counters by $-\vec{v}$.
    This is easily implementable by a path of length at most $B(n)$, for if $\norm{\vec{v}} > B(n)$ this instance is trivially false.
    We then append a path from $r$ to a new state $s$ whose transitions add $B(n)$ to
    every counter. 
    It is easy to see that there is a run from $p(\vec{u})$ to $s(B(n)\cdot\vec{1})$ in $\Vv'$ if and only if there exists a run from $p(\vec{u})$ to $q(\vec{v})$ in $\Vv$. Since $\norm{\Vv'} = \Oh(n + B(n))$, for $B(n) = \gamma \cdot n$ we can ensure that $B(\norm{\Vv'}) = \delta \norm{\Vv'}$ for some $\delta$.
	Thus, $R^{\gamma n}(n) \leq \Oh(C^{\delta n}(\Oh(n))) \leq \Oh(C^{\delta n}(n))$.
\end{proof}

\subsection*{Lower Bounds for Reachability in Linearly-Bounded VASS}

To obtain further lower bounds for the coverability problem in $(\Oh(n),d)$-VASS, by Lemma~\ref{lem:equiv_bounded_reach_cover}, we can equivalently find lower bounds for the reachability problem in $(\Oh(n),d)$-VASS.
In Theorem~\ref{thm:bounded-reachability}, we will assume a well-established hypothesis concerning the time required to find hypercliques in 3-uniform hypergraphs.
In fact, Lincoln, Vassilevska Williams, and Williams state and justify an even stronger hypothesis about $\mu$-uniform hypergraphs for every $\mu \geq 3$~{\cite[Hypothesis 1.4]{LincolnVWW18}}. 
We will use this computational complexity hypothesis to expose precise lower bounds on the time complexity of reachability in linearly-bounded fixed dimension unary VASS.
\begin{hypothesis}[$k$-Hyperclique Hypothesis~{\cite[Hypothesis 1.4]{LincolnVWW18}}]
    \label{hyp:hyperclique}
    Let $k \ge 3$ be an integer. On Word-RAM with $\Oh(\log(n))$ bit words, finding an $k$-hyperclique in a $3$-uniform hypergraph on $n$ vertices requires $n^{k - o(1)}$ time.
\end{hypothesis}

For the remainder of this section, we focus on the proof of the following Theorem.
\begin{theorem}\label{thm:bounded-reachability}
    Assuming Hypothesis~\ref{hyp:hyperclique}, reachability in unary $(\Oh(n), d+2)$-VASS of size $n$ requires $n^{d-o(1)}$ time.
\end{theorem}

The lower bound is obtained via reduction from finding hyperclique in $3$-uniform hypergraphs, hence the lower bound is subject to the $k$-Hyperclique Hypothesis.
We present our reduction in two steps.
The first step is an intermediate step, in Lemma~\ref{lem:hyperclique-reduction} we offer a reduction to an instance of reachability in unary VASS with a limited number of zero-tests.
The second step extends the first, in Lemma~\ref{lem:czerwinski-orlikowski} we modify the reduction by adding a counter so zero-tests are absented.
This extension leverages the recently developed \emph{controlling counter technique} of Czerwi\'nski and Orlikowski~\cite{CzerwinskiO21}.
This technique allows for implicit zero-tests to be performed in the presence of a dedicated counter whose transition effects and reachability condition ensure these implicit zero-tests were indeed performed correctly.

It has been shown that we may as assume that the hypergraph is $\ell$-partite for the $k$-Hyperclique Hypothesis~{\cite[Theorem 3.1]{LincolnVWW18}}.
Thus, we may assume that the vertices can be partitioned into $\ell$ disjoint subsets $V = V_1 \cup \cdots \cup V_\ell$ and all hyperedges contain three vertices from distinct subsets $\set{u, v, w} \in V_i \times V_j \times V_k$ for some $1 \leq i < j < k \leq \ell$. 

\begin{lemma}\label{lem:hyperclique-reduction}
    Let $d \ge 1$ be a fixed integer. 
    Given a $4d$-partite $3$-uniform hypergraph $H = (V_1 \cup \ldots \cup V_{4d}, E)$ with $n$ vertices, there exists a unary $(\Oh(n^{4+o(1)}), d+1)$-VASS with $\Oh(d^3)$ zero-tests $\mathcal{T}$ such that there is a $4d$-hyperclique in $H$ if and only if there is a run from $q_I(\vec{0})$ to $q_F(\vec{v})$, for some $\vec{v} \geq \vec{0}$, in $\mathcal{T}$.
    Moreover, $\mathcal{T}$ can be constructed in $\poly{d} \cdot n^{4+o(1)}$ time.
\end{lemma}
\begin{proof}
    We will re-employ some of the ideas already used in the constructions in the proof of Lemma~\ref{lem:k-clique-reduction}.
	In particular, we will use \MultiplyGadget and \DivideGadget subprocedures, see Algorithm~\ref{alg:multiply-divide}.
	Let us denote the $d+1$ counters $\vr{x}_1, \ldots, \vr{x}_d, \vr{y}$.
	The collective role of $\vr{x}_1, \ldots, \vr{x}_d$ is to maintain a representation of the $4d$ vertices forming the $4d$-hyperclique.
    The role of $\vr{y}$ is to ensure multiplications and divisions are completed correctly.
    Just as previously seen, before the execution of \MultiplyGadget or \DivideGadget, we require $\vr{y}=0$.
    We will combine these subprocedures to construct new subprocedures for unary
	$(d+1)$-VASS with zero-tests to verify properties related to $3$-uniform hypergraphs.

    We start by finding the first $4d\cdot \ell$ primes. 
    We associate a distinct prime $p_v$ to each vertex $v \in V$.
    Now we encode the chosen $4d$ vertices that will form the $4d$-clique by storing, on $d$ many counters, products of four primes corresponding to four of the selected vertices.
	Therefore, after the initial guessing part, the value of counter $\vr{x}_i$ will be $p_t \cdot p_u \cdot p_v \cdot p_w$ for some vertices $t, u, v, w \in V$.
    Roughly speaking, we store the product of four primes on one counter so that
	the maximum observable counter value matches the size of the resulting VASS with zero-tests.

    \begin{algorithm}
    \SetKwInOut{Input}{input}
    \Input{$\vr{x}_1, \ldots, \vr{x}_d, \vr{y} = 0$}
    \BlankLine
    \For{$i\leftarrow 1$ \KwTo $d$}{
    	\increase{\vr{x}_i}{1}\;
        \For{$j \leftarrow 1$ \KwTo $4$}{
        $\textbf{guess } k \in \{1,\ldots,n\}$ \;
    	\Multiply{\vr{x}_i}{p_{k}}\;
    }
    }
    \For{$(i,j,k) \in \{1,\ldots,4d\}^3$, $i \neq j \neq k \neq i$}{
    $\textbf{guess } e \in E \cap (V_i \times V_j \times V_k)$ \;
    \HyperEdge{e}\;
    }
    \BlankLine
    \caption{A counter program representing the VASS with zero tests. 
    As seen earlier, the variables in $\textbf{for}$ are just syntactic sugar for repeating similar lines of code, corresponding to states in the VASS. 
    Similarly, the variables in $\textbf{guess}$ represent non-deterministic branching transitions in a VASS.}
    \label{algo_hyper}
\end{algorithm}
    
    \subparagraph*{Guessing part}
    The VASS presented in Algorithm~\ref{algo_hyper} implements the following algorithm.
	Guess $4d$ vertices, not necessarily distinct, and check whether they form a $4d$-hyperclique. 
    Note that this algorithm is correct because guessing vertices that are the same or in the same $V_i$ does not help us. 
    In contrast to Algorithm~\ref{algo_vass}, the main difference is that the guessed vertices are encoded as quadruple products of primes across counters $\vr{x}_1$, \ldots, $\vr{x}_d$. 
    We do not store the entire product of $4d$ primes explicitly, only the values of each counter.

    \subparagraph*{Checking part}
    In the second part, we verify that we have selected a $4d$-hyperclique by
	testing for each of the ${{4d} \choose 3}$ hyperedges.
    This is achieved essentially in the same way as \Edge{e} was implemented in Algorithm~\ref{algo_edge}.
    We check that between every triplet of vertex subsets there is a hyperedge that has all of the vertices selected in the first part.
    We implement \HyperEdgeGadget subprocedure for checking an individual hyperedge, see Figure~\ref{fig:hyperedge-gadget}.
    This subprocedure checks that the three primes corresponding to the three vertices
	in the hyperedge can divide one of the values stored in $\vr{x}_1, \ldots,
	\vr{x}_d$.
    For ease of presentation and as previously mentioned, we introduce a \VertexSelectedGadget subprocedure that checks whether a given vertex has been selected, see Figure~\ref{fig:vertex-selected-gadget}.
    \begin{figure}[ht!]
        \centering
        \begin{tikzpicture}
	\node (v) at (4, -2) {\VertexSelected{v}};

	\node[fill, circle, inner sep = 0.03in] (q0) at (0,0) {};
	\node[fill, circle, inner sep = 0.03in] (q1) at (8,0) {};

	\node[draw, rectangle, line width = 0.01in, minimum width = 1in] (d1) at (2.25, 1) {\small\Divide{\vr{x}_1}{p_v}};
	\node[draw, rectangle, line width = 0.01in, minimum width = 1in] (m1) at (5.75, 1) {\small\Multiply{\vr{x}_1}{p_v}};
	\node[draw, rectangle, line width = 0.01in, minimum width = 1in] (d2) at (2.25, 0.3) {\small\Divide{\vr{x}_2}{p_v}};
	\node[draw, rectangle, line width = 0.01in, minimum width = 1in] (m2) at (5.75, 0.3) {\small\Multiply{\vr{x}_2}{p_v}};
	\node[rotate=90] at (2.25, -0.35) {$\cdots$};
	\node[rotate=90] at (5.75, -0.35) {$\cdots$};
	\node[draw, rectangle, line width = 0.01in, minimum width = 1in] (dd) at (2.25, -1) {\small\Divide{\vr{x}_d}{p_v}};
	\node[draw, rectangle, line width = 0.01in, minimum width = 1in] (md) at (5.75, -1) {\small\Multiply{\vr{x}_d}{p_v}};

	\path[-stealth, line width = 0.01in, shorten <= 0.02in, shorten >= 0.02in] (q0) edge[bend left = 15] (d1.west);
	\path[-stealth, line width = 0.01in, shorten <= 0.02in, shorten >= 0.02in] (q0) edge[bend left = 8] (d2.west);
	\path[-stealth, line width = 0.01in, shorten <= 0.02in, shorten >= 0.02in] (q0) edge[bend right = 15] (dd.west);

	\draw[-stealth, line width = 0.01in, shorten <= 0.02in, shorten >= 0.02in] (d1.east) -- (m1.west);
	\draw[-stealth, line width = 0.01in, shorten <= 0.02in, shorten >= 0.02in] (d2.east) -- (m2.west);
	\draw[-stealth, line width = 0.01in, shorten <= 0.02in, shorten >= 0.02in] (dd.east) -- (md.west);

	\path[-stealth, line width = 0.01in, shorten <= 0.02in, shorten >= 0.02in] (m1.east) edge[bend left = 15] (q1);
	\path[-stealth, line width = 0.01in, shorten <= 0.02in, shorten >= 0.02in] (m2.east) edge[bend left = 8] (q1);
	\path[-stealth, line width = 0.01in, shorten <= 0.02in, shorten >= 0.02in] (md.east) edge[bend right = 15] (q1);
\end{tikzpicture}
        \caption{The \VertexSelectedGadget subprocedure implemented in a unary $(d+1)$-VASS with zero-tests.
        To instantiate this subprocedure, a vertex $v$ is specified so that the $d$ counters can be checked for divisibility by the prime $p_v$, with the effect of checking whether the vertex $v$ has been selected.
        The counter $y$ is used by the \DivideGadget and \MultiplyGadget subprocedures to ensure the division and multiplications are completed correctly.}
        \label{fig:vertex-selected-gadget}
    \end{figure}
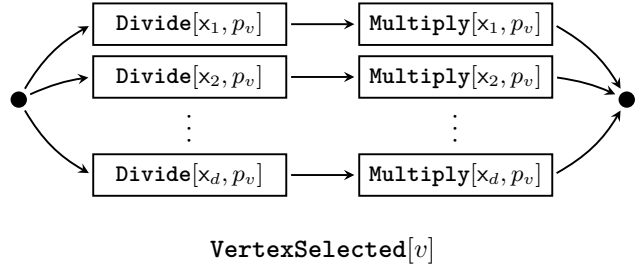
    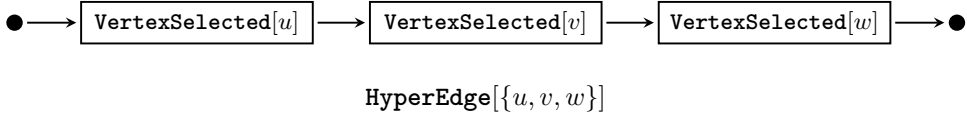
\begin{figure}[ht!]
        \centering
        \begin{tikzpicture}
	\node (e) at (6.2, -1) {\HyperEdge{\set{u, v, w}}};

	\node[fill, circle, inner sep = 0.03in] (p) at (0, 0) {};
	\node[fill, circle, inner sep = 0.03in] (q) at (12.4, 0) {};

	\node[draw, rectangle, line width = 0.01in, minimum width = 1.2in] (su) at (2.4, 0) {\small\VertexSelected{u}};
	\node[draw, rectangle, line width = 0.01in, minimum width = 1.2in] (sv) at (6.2, 0) {\small\VertexSelected{v}};
	\node[draw, rectangle, line width = 0.01in, minimum width = 1.2in] (sw) at (10, 0) {\small\VertexSelected{w}};

	\draw[-stealth, line width = 0.01in, shorten <= 0.02in, shorten >= 0.02in] (p) -- (su.west);
	\draw[-stealth, line width = 0.01in, shorten <= 0.02in, shorten >= 0.02in] (su.east) -- (sv.west);
	\draw[-stealth, line width = 0.01in, shorten <= 0.02in, shorten >= 0.02in] (sv.east) -- (sw.west);
	\draw[-stealth, line width = 0.01in, shorten <= 0.02in, shorten >= 0.02in] (sw.east) -- (q);
\end{tikzpicture}
        \caption{The \HyperEdgeGadget subprocedure implemented in a unary $(d+1)$-VASS with zero-tests.
        Note that the three vertices of the given hyperedge may be stored across
		any of the $d$ counters $\vr{x}_1, \ldots, \vr{x}_d$.
        Therefore, we make use of the \VertexSelectedGadget subprocedure three times to check if indeed $u$, $v$, and $w$ have been selected.}
        \label{fig:hyperedge-gadget}
    \end{figure}

    Now, we use the aforementioned subprocedures to construct the checking part of $\mathcal{T}$. 
    This part consists of a sequence of ${{4d} \choose 3}$ non-deterministic branching sections, one for each triplet of vertex subsets $U \times V \times W$.
    In each branching section, there is an instance of the \HyperEdgeGadget subprocedure for each of the hyperedges $\set{u, v, w} \in U \times V \times W$. 
    In order to reach the final state $q_F$, there must be a hyperedge between each of the $4d$ vertices selected in the first part of $\mathcal{T}$.
    Thus, there is a $4d$-hyperclique in $H$ if and only if there is a run from $q_I(\vec{0})$ to $q_F(\vec{v})$ for some $\vec{v} \geq \vec{0}$ in $\mathcal{T}$.
    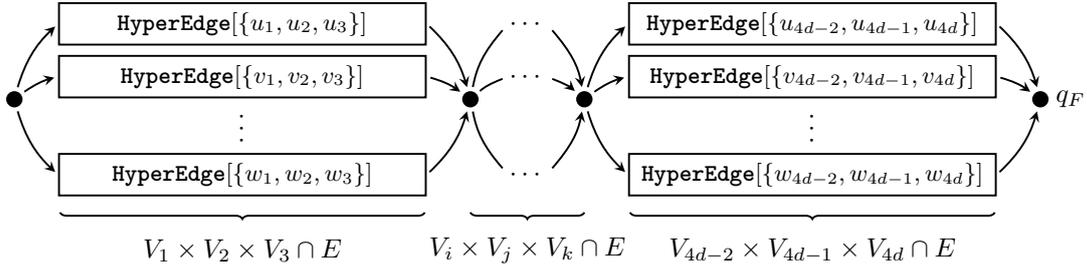
\begin{figure}[ht!]
        \centering
        \begin{tikzpicture}
	\node[fill, circle, inner sep = 0.03in] (q1) at (0, 0) {};
	\node[fill, circle, inner sep = 0.03in] (q2) at (6, 0) {};
	\node[fill, circle, inner sep = 0.03in] (q3) at (7.5, 0) {};
	\node[fill, circle, inner sep = 0.03in] (q4) at (13.5, 0) {};
		\node (q) at (13.9, 0) {$q_F$};

	\node[draw, rectangle, line width = 0.01in, minimum width = 1.9in] (g11) at (3, 1) {\small\HyperEdge{\set{u_1, u_2, u_3}}};
	\node[draw, rectangle, line width = 0.01in, minimum width = 1.9in] (g12) at (3, 0.3) {\small\HyperEdge{\set{v_1, v_2, v_3}}};
	\node[rotate = 90] (g1d) at (3, -0.35) {$\cdots$};
	\node[draw, rectangle, line width = 0.01in, minimum width = 1.9in] (g1n) at (3, -1) {\small\HyperEdge{\set{w_1, w_2, w_3}}};
	\draw[line width = 0.01in, decoration={brace, mirror}, decorate] (0.6, -1.5) -- (5.4, -1.5);
	\node at (3, -2) {$V_1 \times V_2 \times V_3 \cap E$};

	\node (g21) at (6.5, 1) {};
	\node (g22) at (6.5, 0.3) {};
	\node (g2n) at (6.5, -1) {};
	\node (qd) at (6.75, 1) {$\cdots$};
	\node (qd) at (6.75, 0.3) {$\cdots$};
	\node (qd) at (6.75, -1) {$\cdots$};
	\node (g31) at (7, 1) {};
	\node (g32) at (7, 0.3) {};
	\node (g3n) at (7, -1) {};
	\draw[line width = 0.01in, decoration={brace, mirror}, decorate] (6, -1.5) -- (7.5, -1.5);
	\node at (6.75, -2) {$V_{i} \times V_{j} \times V_{k} \cap E$};

	\node[draw, rectangle, line width = 0.01in, minimum width = 1.9in] (gk1) at (10.5, 1) {\small\HyperEdge{\set{u_{4d-2}, u_{4d-1}, u_{4d}}}};
	\node[draw, rectangle, line width = 0.01in, minimum width = 1.9in] (gk2) at (10.5, 0.3) {\small\HyperEdge{\set{v_{4d-2}, v_{4d-1}, v_{4d}}}};
	\node[rotate = 90] (gkd) at (10.5, -0.35) {$\cdots$};
	\node[draw, rectangle, line width = 0.01in, minimum width = 1.9in] (gkn) at (10.5, -1) {\small\HyperEdge{\set{w_{4d-2}, w_{4d-1}, w_{4d}}}};
	\draw[line width = 0.01in, decoration={brace, mirror}, decorate] (8.1, -1.5) -- (12.9, -1.5);
	\node at (10.5, -2) {$V_{4d-2} \times V_{4d-1} \times V_{4d} \cap E$};

	\path[-stealth, line width = 0.01in, shorten <= 0.02in, shorten >= 0.02in] (q1) edge[bend left = 15] (g11.west);
	\path[-stealth, line width = 0.01in, shorten <= 0.02in, shorten >= 0.02in] (g11.east) edge[bend left = 15] (q2);
	\path[-stealth, line width = 0.01in, shorten <= 0.02in, shorten >= 0.02in] (q1) edge[bend left = 10] (g12.west);
	\path[-stealth, line width = 0.01in, shorten <= 0.02in, shorten >= 0.02in] (g12.east) edge[bend left = 10] (q2);
	\path[-stealth, line width = 0.01in, shorten <= 0.02in, shorten >= 0.02in] (q1) edge[bend right = 15] (g1n.west);
	\path[-stealth, line width = 0.01in, shorten <= 0.02in, shorten >= 0.02in] (g1n.east) edge[bend right = 15] (q2);

	\path[line width = 0.01in, shorten <= 0.02in, shorten >= -0.02in] (q2) edge[bend left = 15] (g21);
	\path[line width = 0.01in, shorten <= 0.02in, shorten >= -0.02in] (q2) edge[bend left = 10] (g22);
	\path[line width = 0.01in, shorten <= 0.02in, shorten >= -0.02in] (q2) edge[bend right = 15] (g2n);
	\path[-stealth, line width = 0.01in, shorten <= -0.02in, shorten >= 0.02in] (g31) edge[bend left = 15] (q3);
	\path[-stealth, line width = 0.01in, shorten <= -0.02in, shorten >= 0.02in] (g32) edge[bend left = 10] (q3);
	\path[-stealth, line width = 0.01in, shorten <= -0.02in, shorten >= 0.02in] (g3n) edge[bend right = 15] (q3);

	\path[-stealth, line width = 0.01in, shorten <= 0.02in, shorten >= 0.02in] (q3) edge[bend left = 15] (gk1.west);
	\path[-stealth, line width = 0.01in, shorten <= 0.02in, shorten >= 0.02in] (gk1.east) edge[bend left = 15] (q4);
	\path[-stealth, line width = 0.01in, shorten <= 0.02in, shorten >= 0.02in] (q3) edge[bend left = 10] (gk2.west);
	\path[-stealth, line width = 0.01in, shorten <= 0.02in, shorten >= 0.02in] (gk2.east) edge[bend left = 10] (q4);
	\path[-stealth, line width = 0.01in, shorten <= 0.02in, shorten >= 0.02in] (q3) edge[bend right = 15] (gkn.west);
	\path[-stealth, line width = 0.01in, shorten <= 0.02in, shorten >= 0.02in] (gkn.east) edge[bend right = 15] (q4);
\end{tikzpicture}
        \caption{The check part of the unary $(d+1)$-VASS with zero-tests $\mathcal{T}$ for finding a $4d$-hyperclique in $4d$-partite hypergraph.}
        \label{fig:hyperclique-check-part}
    \end{figure}

    We are now able to finalize the proof. First, we will carefully analyse the maximum counter value observed on any run and count the number of zero-tests performed on any run. Then, we will evaluate the size of $\mathcal{T}$.
    
	The highest counter value observed by each counter $\vr{x}_1, \ldots, \vr{x}_d$ is the product of four primes.
	The highest counter value observed by $\vr{y}$ is equal to the highest counter
	value observed by any of the other counters  $\vr{x}_1, \ldots, \vr{x}_d$.
    Therefore the bound on the highest value observed, altogether can be bounded about by $p^4$ where $p$ is the $n$-th prime.
    By the Prime Number Theorem (for example, see~\cite{Zagier79}) we know that $p \in \Oh(n \log(n))$.
    Therefore, every run from $q_I(\vec{0})$ to $q_F(\vec{v})$, for some $\vec{v} \geq \vec{0}$, in $\mathcal{T}$ is $\Oh(n^{4+o(1)})$-bounded.

    Zero-tests are only performed by the \MultiplyGadget and \DivideGadget subprocedures, each instance of these subprocedures contains two zero-tests.
    It remains to count the number of these subprocedures are executed on any run.
    In the guessing part, there is one \MultiplyGadget subprocedure for each of the $4d$ vertices selected to form a $4d$-hyperclique.
    In the checking part, there is a sequence of ${{4d} \choose 3}$ many \HyperEdgeGadget subprocedures.
    Each \HyperEdgeGadget subprocedure contains three instances of the \VertexSelectedGadget subprocedure which executes one \DivideGadget subprocedure and one \MultiplyGadget subprocedure.
    In total, there are $2(4d + 6{{4d} \choose 3}) \in \Oh(d^3)$ many zero-tests are performed in any run from $q_I(\vec{0})$ to $q_F(\vec{v})$, for some $\vec{v} \geq \vec{0}$, in $\mathcal{T}$.

    Finally, each instance of the \MultiplyGadget and \DivideGadget subprocedures has size $\Oh(n\log{n})$.
    Note that the first $n$ primes can be found in $\Oh(n^{1+o(1)})$-time~\cite{AgrawalKS04}.
    In the guessing part, there are $4dn$ instances of the \MultiplyGadget subprocedure.
    In the checking part, there is an instance of the \HyperEdgeGadget subprocedure for each edge in the hypergraph.
    The \HyperEdgeGadget subprocedures themselves appear in ${{4d} \choose 3}$ many collections, one for each triplet of vertex subsets.
    Each \HyperEdgeGadget subprocedure contains three instances of the \VertexSelectedGadget subprocedure, which contains $d$ instances of the \MultiplyGadget subprocedure and $d$ instances of the \DivideGadget subprocedure.
    Therefore, in total $\mathcal{T}$ has polynomial size and can be constructed in time $\Oh(dn \cdot n\log{n} + m \cdot \binom{4d}{3} \cdot d \cdot n\log{n})$, where $m \in \Oh(n^3)$ is the total number of hyperedges.
\end{proof}

Consider our earlier described two-step approach towards proving Theorem~\ref{thm:bounded-reachability} by first obtaining a unary $(d+1)$-VASS $\mathcal{T}$ with zero-tests, then obtaining a unary $(d+2)$-VASS $\Vv$ by increasing the dimension by one and removing the zero-tests.
In actuality, both steps occur together to prove Theorem~\ref{thm:bounded-reachability}.
Ultimately, the $d+2$ counters of $\mathcal{V}$ have the following roles.
The counters $\vr{x}_1, \ldots, \vr{x}_d$  are used to store the products of primes
corresponding to vertices of hyperclique.
The counter $\vr{y}$ is used to complete multiplications and divisions.
In the remainder of this section, we add the $d+2$-nd counter that is used to ensure
the (implicit) zero-tests are performed faithfully. 
We achieve this by leveraging the \emph{controlling counter technique} introduced by Czerwi\'nski and Orlikowski~\cite{CzerwinskiO21}. 
The following is the restatement of their technique, their lemma has been restricted to our scenario and rewritten using the notation of this paper.

\begin{lemma}[{\cite[Lemma 10]{CzerwinskiO21}}] \label{lem:czerwinski-orlikowski}
    Let $\rho$ be a run in a $(d+2)$-VASS such that $\run{q_I(\vec{0})}{\rho}{q_F(\vec{0})}$.
    Further, let $q_0(\vec{v}_0), q_1(\vec{v}_1) \ldots, q_r(\vec{v}_r)$ be some distinguished configurations observed along the run $\rho$ with $q_0(\vec{v}_0) = q_I(\vec{0})$ and $q_r(\vec{v}_r) = q_F(\vec{0})$ and let $\rho_j$ be the segment of $\rho$ that is between $q_{j-1}(\vec{v}_{j-1})$ and $q_j(\vec{v}_j)$, so $\rho$ can be described as
    \begin{equation*}
        q_I(\vec{0}) = q_0(\vec{v}_1) \xrightarrow{\rho_1} q_1(\vec{v}_1) \rightarrow \cdots \rightarrow  q_{r-1}(\vec{v}_{r-1}) \xrightarrow{\rho_r} q_r(\vec{v}_r) = q_F(\vec{0}).
    \end{equation*}
    Let $S_1, \ldots, S_d, S_{d+1} \subseteq \set{0, 1, \ldots, r}$ be the sets
	of indices of the distinguished configurations where zero-tests could be
	performed on counters $\vr{x}_1, \ldots, \vr{x}_d, \vr{x}_{d+1}$, respectively.
    Let $t_{j,i} = \abs{\set{ s \geq j : s \in S_i }}$ be the number of
	zero-test for the counter $\vr{x}_i$ in the remainder of the run $\rho_{j+1}\cdots\rho_r$.
    Given that $\vec{v}_0 = \vec{0}$ and $\vec{v}_r = \vec{0}$, if 
    \begin{equation}
		\eff{\rho_j}[d+2] = \sum_{i=1}^{d+1} t_{j, i}\cdot \eff{\rho_j}[i],
		\label{equ:controlling-counter-effect}
    \end{equation}
    then for every $i \in \set{1, \ldots, d, d+1}$ and $j \in S_i$, we know that $\vec{v}_j[i] = 0$.
\end{lemma}

With Lemma~\ref{lem:czerwinski-orlikowski} in hand we can ensure that every zero-test is
executed correctly and conclude this section with a proof of Theorem~\ref{thm:bounded-reachability}.

\begin{proof}[Proof of Theorem~\ref{thm:bounded-reachability}]
    Consider the reduction, presented in Lemma~\ref{lem:hyperclique-reduction}, from finding a $4d$-hyperclique in a $4d$-partite $3$-uniform hypergraph $H$ to reachability in $(\Oh(n^{4+o(1)}), d+1)$-VASS with $\Oh(d^3)$ zero-tests.
    Now, given Lemma~\ref{lem:czerwinski-orlikowski}, we will add a controlling counter to $\mathcal{T}$ so that the zero-tests on the $d+1$ counters $\vr{x}_1, \ldots, \vr{x}_d, \vr{y}$ are instead performed implicitly.
	So we introduce another counter $\vr{z}$ that receives updates on transitions, consistent with Equation~\ref{equ:controlling-counter-effect}, whenever any of the other counters are updated.
	Note that counters $\vr{y}$ and $\vr{z}$, for the sake of a succinct and consistent description, are respectively referred to as counters $\vr{x}_{d+1}$ and $\vr{x}_{d+2}$ in the statement of Lemma~\ref{lem:czerwinski-orlikowski}. 
    Moreover, notice that the maximum value of $\vr{z}$ is bounded by $\poly{d} \cdot \left(\sum_{i=1}^{d+1} \vr{x}_i\right) \in \poly{d} \cdot n^{4+o(1)}$.

	Therefore, we have constructed a unary $(\poly{d} \cdot n^{4+o(1)}, d+2)$-VASS $\mathcal{V}$ with the property that there $H$ contains a $4d$-hyperclique if and only if there is a run from $q_I'(\vec{0})$ to $q_F'(\vec{0})$ in $\mathcal{V}$.
	Such a $(\poly{d}\cdot n^{4+o(1)}, d+2)$-VASS $\mathcal{V}$ has size $\Oh(t\cdot\abs{\mathcal{T}})$ where $t \in \poly{d}$ is the number of zero-tests performed on the run from $q_I(\vec{0})$ to $q_F(\vec{0})$ in $\mathcal{T}$.
	Moreover, $\mathcal{V}$ can be constructed in $\poly{d}\cdot n^{4+o(1)}$ time.
	Hence, if reachability in $(\Oh(n), d+2)$-VASS of size $n$ can be solved faster than $n^{d-o(1)}$, then one can find a $4d$-hyperclique in a $3$-uniform
	hypergraph faster than $n^{4d-o(1)}$, contradicting Hypothesis~\ref{hyp:hyperclique}.
\end{proof}

\section{Conclusion} \label{sec:conclusion}

\subparagraph*{Summary}
In this paper, we have revisited a classical problem of coverability in $d$-VASS. 
We have closed the gap left by Rosier and Yen~\cite{RosierY86} on the length of runs witnessing instances of coverability in $d$-VASS.
We have lowered the upper bound of $n^{2^{\Oh(d \log d)}}$, from Rackoff's technique~\cite{Rackoff78}, to $n^{2^{\Oh(d)}}$ (Theorem~\ref{thm:min-run}), matching the $n^{2^{\Omega(d)}}$ lower bound from Lipton's construction~\cite{Lipton76}.
This accordingly closes the gap on the exact space required for the coverability problem and yields a deterministic $n^{2^{\Oh(d)}}$-time algorithm for coverability in $d$-VASS (Corollary~\ref{cor:coverability-algorithm}).
We complement this with a matching lower bound conditional on ETH; there does not exist a deterministic $n^{2^{o(d)}}$-time algorithm for coverability (Theorem~\ref{thm:unary-kvass-coverability-lb}).
By and large, this settles the exact space and time complexity of coverability in VASS.

In addition, we study linearly-bounded unary $d$-VASS.
Here, coverability and reachability are equivalent and the trivial exhaustive search $\Oh(n^{d+1})$ algorithm is near-optimal.
We prove that reachability in linearly-bounded 1-VASS requires $n^{2-o(1)}$-time under the $k$-cycle hypothesis (Theorem~\ref{thm:k-cycle-to-vass}), matching the trivial upper bound.
We further prove that reachability in linearly-bounded $(d+2)$-VASS requires $n^{d-o(1)}$ time under the $3$-uniform hyperclique hypothesis (Theorem~\ref{thm:bounded-reachability}).

\subparagraph*{Open Problems}
The \emph{boundedness problem}, a problem closely related to coverability, asks whether, from a given initial configuration, the set of all reachable configurations is finite.
This problem was also studied by Lipton then Rackoff and is \class{EXPSPACE}-complete~\cite{Lipton76, Rackoff78}.
Boundedness was further analysed by Rosier and Yen~{\cite[Theorem 2.1]{RosierY86}} and the same gap also exists for the exact space required. 
We leave the same improvement, to eliminate the same twice-exponentiated $\log(d)$ factor, as an open problem.

Our lower bounds for the time complexity of coverability and reachability in linearly-bounded unary $d$-VASS, for $d \geq 2$, leave a gap of up to $n^{3+o(1)}$, see Table~\ref{tab:coverability-time-bounds}.
We leave it as an open problem to either improve upon the upper bound $\Oh(n^{d+1})$ given by the trivial algorithm, or to raise our conditional lower bounds.

\bibliography{bib}

\begin{thebibliography}{10}

\bibitem{AbdullaCJT00}
Parosh~Aziz Abdulla, Karlis Cerans, Bengt Jonsson, and Yih{-}Kuen Tsay.
\newblock Algorithmic analysis of programs with well quasi-ordered domains.
\newblock {\em Inf. Comput.}, 160(1-2):109--127, 2000.
\newblock \href {https://doi.org/10.1006/inco.1999.2843}
  {\path{doi:10.1006/inco.1999.2843}}.

\bibitem{AgrawalKS04}
Manindra Agrawal, Neeraj Kayal, and Nitin Saxena.
\newblock Primes is in {P}.
\newblock {\em Annals of mathematics}, pages 781--793, 2004.

\bibitem{AlmagorCPSW20}
Shaull Almagor, Nathann Cohen, Guillermo~A. P{\'{e}}rez, Mahsa Shirmohammadi,
  and James Worrell.
\newblock {C}overability in 1-{VASS} with {D}isequality {T}ests.
\newblock In Igor Konnov and Laura Kov{\'{a}}cs, editors, {\em 31st
  International Conference on Concurrency Theory, {CONCUR} 2020, September 1-4,
  2020, Vienna, Austria (Virtual Conference)}, volume 171 of {\em LIPIcs},
  pages 38:1--38:20. Schloss Dagstuhl - Leibniz-Zentrum f{\"{u}}r Informatik,
  2020.
\newblock \href {https://doi.org/10.4230/LIPIcs.CONCUR.2020.38}
  {\path{doi:10.4230/LIPIcs.CONCUR.2020.38}}.

\bibitem{AlonYZ97}
Noga Alon, Raphael Yuster, and Uri Zwick.
\newblock Finding and counting given length cycles.
\newblock {\em Algorithmica}, 17(3):209--223, 1997.
\newblock \href {https://doi.org/10.1007/BF02523189}
  {\path{doi:10.1007/BF02523189}}.

\bibitem{AnconaHRWW19}
Bertie Ancona, Monika Henzinger, Liam Roditty, Virginia~Vassilevska Williams,
  and Nicole Wein.
\newblock Algorithms and hardness for diameter in dynamic graphs.
\newblock In Christel Baier, Ioannis Chatzigiannakis, Paola Flocchini, and
  Stefano Leonardi, editors, {\em 46th International Colloquium on Automata,
  Languages, and Programming, {ICALP} 2019, July 9-12, 2019, Patras, Greece},
  volume 132 of {\em LIPIcs}, pages 13:1--13:14. Schloss Dagstuhl -
  Leibniz-Zentrum f{\"{u}}r Informatik, 2019.
\newblock \href {https://doi.org/10.4230/LIPIcs.ICALP.2019.13}
  {\path{doi:10.4230/LIPIcs.ICALP.2019.13}}.

\bibitem{BlondinEFGHLMT21}
Michael Blondin, Matthias Englert, Alain Finkel, Stefan G{\"{o}}ller, Christoph
  Haase, Ranko Lazi\'c, Pierre McKenzie, and Patrick Totzke.
\newblock {T}he {R}eachability {P}roblem for {T}wo-{D}imensional {V}ector
  {A}ddition {S}ystems with {S}tates.
\newblock {\em J. {ACM}}, 68(5):34:1--34:43, 2021.
\newblock \href {https://doi.org/10.1145/3464794} {\path{doi:10.1145/3464794}}.

\bibitem{BlondinFHH16}
Michael Blondin, Alain Finkel, Christoph Haase, and Serge Haddad.
\newblock {Approaching the Coverability Problem Continuously}.
\newblock In Marsha Chechik and Jean{-}Fran{\c{c}}ois Raskin, editors, {\em
  Tools and Algorithms for the Construction and Analysis of Systems - 22nd
  International Conference, {TACAS} 2016, Held as Part of the European Joint
  Conferences on Theory and Practice of Software, {ETAPS} 2016, Eindhoven, The
  Netherlands, April 2-8, 2016, Proceedings}, volume 9636 of {\em Lecture Notes
  in Computer Science}, pages 480--496. Springer, 2016.
\newblock \href {https://doi.org/10.1007/978-3-662-49674-9\_28}
  {\path{doi:10.1007/978-3-662-49674-9\_28}}.

\bibitem{BlondinHO21}
Michael Blondin, Christoph Haase, and Philip Offtermatt.
\newblock Directed reachability for infinite-state systems.
\newblock In Jan~Friso Groote and Kim~Guldstrand Larsen, editors, {\em Tools
  and Algorithms for the Construction and Analysis of Systems - 27th
  International Conference, {TACAS} 2021, Held as Part of the European Joint
  Conferences on Theory and Practice of Software, {ETAPS} 2021, Luxembourg
  City, Luxembourg, March 27 - April 1, 2021, Proceedings, Part {II}}, volume
  12652 of {\em Lecture Notes in Computer Science}, pages 3--23. Springer,
  2021.
\newblock \href {https://doi.org/10.1007/978-3-030-72013-1\_1}
  {\path{doi:10.1007/978-3-030-72013-1\_1}}.

\bibitem{BojanczykDMSS11}
Miko\l{}aj Boja\'nczyk, Claire David, Anca Muscholl, Thomas Schwentick, and Luc
  Segoufin.
\newblock Two-variable logic on data words.
\newblock {\em {ACM} Trans. Comput. Log.}, 12(4):27:1--27:26, 2011.
\newblock \href {https://doi.org/10.1145/1970398.1970403}
  {\path{doi:10.1145/1970398.1970403}}.

\bibitem{BozzelliG11}
Laura Bozzelli and Pierre Ganty.
\newblock Complexity analysis of the backward coverability algorithm for
  {VASS}.
\newblock In Giorgio Delzanno and Igor Potapov, editors, {\em Reachability
  Problems - 5th International Workshop, {RP} 2011, Genoa, Italy, September
  28-30, 2011. Proceedings}, volume 6945 of {\em Lecture Notes in Computer
  Science}, pages 96--109. Springer, 2011.
\newblock \href {https://doi.org/10.1007/978-3-642-24288-5\_10}
  {\path{doi:10.1007/978-3-642-24288-5\_10}}.

\bibitem{ChenCFHJKX05}
Jianer Chen, Benny Chor, Mike Fellows, Xiuzhen Huang, David~W. Juedes, Iyad~A.
  Kanj, and Ge~Xia.
\newblock Tight lower bounds for certain parameterized {NP}-hard problems.
\newblock {\em Inf. Comput.}, 201(2):216--231, 2005.
\newblock \href {https://doi.org/10.1016/j.ic.2005.05.001}
  {\path{doi:10.1016/j.ic.2005.05.001}}.

\bibitem{ChenHKX06}
Jianer Chen, Xiuzhen Huang, Iyad~A. Kanj, and Ge~Xia.
\newblock Strong computational lower bounds via parameterized complexity.
\newblock {\em J. Comput. Syst. Sci.}, 72(8):1346--1367, 2006.
\newblock \href {https://doi.org/10.1016/j.jcss.2006.04.007}
  {\path{doi:10.1016/j.jcss.2006.04.007}}.

\bibitem{ComonJ98}
Hubert Comon and Yan Jurski.
\newblock Multiple counters automata, safety analysis and presburger
  arithmetic.
\newblock In Alan~J. Hu and Moshe~Y. Vardi, editors, {\em Computer Aided
  Verification, 10th International Conference, {CAV} '98, Vancouver, BC,
  Canada, June 28 - July 2, 1998, Proceedings}, volume 1427 of {\em Lecture
  Notes in Computer Science}, pages 268--279. Springer, 1998.
\newblock \href {https://doi.org/10.1007/BFb0028751}
  {\path{doi:10.1007/BFb0028751}}.

\bibitem{CyganFKLMPPS15}
Marek Cygan, Fedor~V. Fomin, \L{}ukasz Kowalik, Daniel Lokshtanov, D{\'{a}}niel
  Marx, Marcin Pilipczuk, Micha\l{} Pilipczuk, and Saket Saurabh.
\newblock {\em Parameterized Algorithms}.
\newblock Springer, 2015.
\newblock \href {https://doi.org/10.1007/978-3-319-21275-3}
  {\path{doi:10.1007/978-3-319-21275-3}}.

\bibitem{Czerwinski0LLM20}
Wojciech Czerwi\'nski, Slawomir Lasota, Ranko Lazic, J{\'{e}}r{\^{o}}me Leroux,
  and Filip Mazowiecki.
\newblock {R}eachability in {F}ixed {D}imension {V}ector {A}ddition {S}ystems
  with {S}tates.
\newblock In Igor Konnov and Laura Kov{\'{a}}cs, editors, {\em 31st
  International Conference on Concurrency Theory, {CONCUR} 2020, September 1-4,
  2020, Vienna, Austria (Virtual Conference)}, volume 171 of {\em LIPIcs},
  pages 48:1--48:21. Schloss Dagstuhl - Leibniz-Zentrum f{\"{u}}r Informatik,
  2020.
\newblock \href {https://doi.org/10.4230/LIPIcs.CONCUR.2020.48}
  {\path{doi:10.4230/LIPIcs.CONCUR.2020.48}}.

\bibitem{CzerwinskiO21}
Wojciech Czerwi\'nski and \L{}ukasz Orlikowski.
\newblock {R}eachability in {V}ector {A}ddition {S}ystems is
  {A}ckermann-complete.
\newblock In {\em 62nd {IEEE} Annual Symposium on Foundations of Computer
  Science, {FOCS} 2021, Denver, CO, USA, February 7-10, 2022}, pages
  1229--1240. {IEEE}, 2021.
\newblock \href {https://doi.org/10.1109/FOCS52979.2021.00120}
  {\path{doi:10.1109/FOCS52979.2021.00120}}.

\bibitem{CzerwinskiO22}
Wojciech Czerwi\'nski and \L{}ukasz Orlikowski.
\newblock {L}ower {B}ounds for the {R}eachability {P}roblem in {F}ixed
  {D}imensional {VASS}es.
\newblock In Christel Baier and Dana Fisman, editors, {\em {LICS} '22: 37th
  Annual {ACM/IEEE} Symposium on Logic in Computer Science, Haifa, Israel,
  August 2 - 5, 2022}, pages 40:1--40:12. {ACM}, 2022.
\newblock \href {https://doi.org/10.1145/3531130.3533357}
  {\path{doi:10.1145/3531130.3533357}}.

\bibitem{DalirrooyfardJVWW22}
Mina Dalirrooyfard, Ce~Jin, Virginia~Vassilevska Williams, and Nicole Wein.
\newblock {Approximation Algorithms and Hardness for n-Pairs Shortest Paths and
  All-Nodes Shortest Cycles}.
\newblock In {\em 63rd {IEEE} Annual Symposium on Foundations of Computer
  Science, {FOCS} 2022, Denver, CO, USA, October 31 - November 3, 2022}, pages
  290--300. {IEEE}, 2022.
\newblock \href {https://doi.org/10.1109/FOCS54457.2022.00034}
  {\path{doi:10.1109/FOCS54457.2022.00034}}.

\bibitem{DalirrooyfardVW21}
Mina Dalirrooyfard, Thuy~Duong Vuong, and Virginia~Vassilevska Williams.
\newblock Graph pattern detection: Hardness for all induced patterns and faster
  noninduced cycles.
\newblock {\em {SIAM} J. Comput.}, 50(5):1627--1662, 2021.
\newblock \href {https://doi.org/10.1137/20M1335054}
  {\path{doi:10.1137/20M1335054}}.

\bibitem{Esparza96}
Javier Esparza.
\newblock {D}ecidability and {C}omplexity of {P}etri {N}et {P}roblems - {A}n
  {I}ntroduction.
\newblock In Wolfgang Reisig and Grzegorz Rozenberg, editors, {\em Lectures on
  Petri Nets {I:} Basic Models, Advances in Petri Nets, the volumes are based
  on the Advanced Course on Petri Nets, held in Dagstuhl, September 1996},
  volume 1491 of {\em Lecture Notes in Computer Science}, pages 374--428.
  Springer, 1996.
\newblock \href {https://doi.org/10.1007/3-540-65306-6\_20}
  {\path{doi:10.1007/3-540-65306-6\_20}}.

\bibitem{EsparzaLMMN14}
Javier Esparza, Rusl{\'{a}}n Ledesma{-}Garza, Rupak Majumdar, Philipp~J. Meyer,
  and Filip Niksic.
\newblock An {SMT}-{B}ased {A}pproach to {C}overability {A}nalysis.
\newblock In Armin Biere and Roderick Bloem, editors, {\em Computer Aided
  Verification - 26th International Conference, {CAV} 2014, Held as Part of the
  Vienna Summer of Logic, {VSL} 2014, Vienna, Austria, July 18-22, 2014.
  Proceedings}, volume 8559 of {\em Lecture Notes in Computer Science}, pages
  603--619. Springer, 2014.
\newblock \href {https://doi.org/10.1007/978-3-319-08867-9\_40}
  {\path{doi:10.1007/978-3-319-08867-9\_40}}.

\bibitem{FearnleyJ13}
John Fearnley and Marcin Jurdzi\'nski.
\newblock {R}eachability in {T}wo-{C}lock {T}imed {A}utomata {I}s
  {PSPACE}-{C}omplete.
\newblock In Fedor~V. Fomin, Rusins Freivalds, Marta~Z. Kwiatkowska, and David
  Peleg, editors, {\em Automata, Languages, and Programming - 40th
  International Colloquium, {ICALP} 2013, Riga, Latvia, July 8-12, 2013,
  Proceedings, Part {II}}, volume 7966 of {\em Lecture Notes in Computer
  Science}, pages 212--223. Springer, 2013.
\newblock \href {https://doi.org/10.1007/978-3-642-39212-2\_21}
  {\path{doi:10.1007/978-3-642-39212-2\_21}}.

\bibitem{FigueiraFSS11}
Diego Figueira, Santiago Figueira, Sylvain Schmitz, and Philippe Schnoebelen.
\newblock Ackermannian and {P}rimitive-{R}ecursive {B}ounds with {D}ickson's
  {L}emma.
\newblock In {\em Proceedings of the 26th Annual {IEEE} Symposium on Logic in
  Computer Science, {LICS} 2011, June 21-24, 2011, Toronto, Ontario, Canada},
  pages 269--278. {IEEE} Computer Society, 2011.
\newblock \href {https://doi.org/10.1109/LICS.2011.39}
  {\path{doi:10.1109/LICS.2011.39}}.

\bibitem{GantyM12}
Pierre Ganty and Rupak Majumdar.
\newblock Algorithmic verification of asynchronous programs.
\newblock {\em {ACM} Trans. Program. Lang. Syst.}, 34(1):6:1--6:48, 2012.
\newblock \href {https://doi.org/10.1145/2160910.2160915}
  {\path{doi:10.1145/2160910.2160915}}.

\bibitem{GermanS92}
Steven~M. German and A.~Prasad Sistla.
\newblock Reasoning about systems with many processes.
\newblock {\em J. {ACM}}, 39(3):675--735, 1992.
\newblock \href {https://doi.org/10.1145/146637.146681}
  {\path{doi:10.1145/146637.146681}}.

\bibitem{HaaseKOW09}
Christoph Haase, Stephan Kreutzer, Jo{\"{e}}l Ouaknine, and James Worrell.
\newblock {R}eachability in {S}uccinct and {P}arametric {O}ne-{C}ounter
  {A}utomata.
\newblock In Mario Bravetti and Gianluigi Zavattaro, editors, {\em {CONCUR}
  2009 - Concurrency Theory, 20th International Conference, {CONCUR} 2009,
  Bologna, Italy, September 1-4, 2009. Proceedings}, volume 5710 of {\em
  Lecture Notes in Computer Science}, pages 369--383. Springer, 2009.
\newblock \href {https://doi.org/10.1007/978-3-642-04081-8\_25}
  {\path{doi:10.1007/978-3-642-04081-8\_25}}.

\bibitem{HaaseOW12}
Christoph Haase, Jo{\"{e}}l Ouaknine, and James Worrell.
\newblock On the relationship between reachability problems in timed and
  counter automata.
\newblock In Alain Finkel, J{\'{e}}r{\^{o}}me Leroux, and Igor Potapov,
  editors, {\em Reachability Problems - 6th International Workshop, {RP} 2012,
  Bordeaux, France, September 17-19, 2012. Proceedings}, volume 7550 of {\em
  Lecture Notes in Computer Science}, pages 54--65. Springer, 2012.
\newblock \href {https://doi.org/10.1007/978-3-642-33512-9\_6}
  {\path{doi:10.1007/978-3-642-33512-9\_6}}.

\bibitem{HopcroftP79}
John~E. Hopcroft and Jean{-}Jacques Pansiot.
\newblock {O}n the {R}eachability {P}roblem for 5-{D}imensional {V}ector
  {A}ddition {S}ystems.
\newblock {\em Theor. Comput. Sci.}, 8:135--159, 1979.
\newblock \href {https://doi.org/10.1016/0304-3975(79)90041-0}
  {\path{doi:10.1016/0304-3975(79)90041-0}}.

\bibitem{ImpagliazzoP01}
Russell Impagliazzo and Ramamohan Paturi.
\newblock On the {C}omplexity of k-{SAT}.
\newblock {\em J. Comput. Syst. Sci.}, 62(2):367--375, 2001.
\newblock \href {https://doi.org/10.1006/jcss.2000.1727}
  {\path{doi:10.1006/jcss.2000.1727}}.

\bibitem{KoppenhagenM00}
Ulla Koppenhagen and Ernst~W. Mayr.
\newblock Optimal algorithms for the coverability, the subword, the
  containment, and the equivalence problems for commutative semigroups.
\newblock {\em Inf. Comput.}, 158(2):98--124, 2000.
\newblock \href {https://doi.org/10.1006/inco.1999.2812}
  {\path{doi:10.1006/inco.1999.2812}}.

\bibitem{LazicS21}
Ranko Lazic and Sylvain Schmitz.
\newblock The ideal view on {R}ackoff's coverability technique.
\newblock {\em Inf. Comput.}, 277:104582, 2021.
\newblock \href {https://doi.org/10.1016/j.ic.2020.104582}
  {\path{doi:10.1016/j.ic.2020.104582}}.

\bibitem{Leroux13}
J{\'{e}}r{\^{o}}me Leroux.
\newblock Vector addition system reversible reachability problem.
\newblock {\em Log. Methods Comput. Sci.}, 9(1), 2013.
\newblock \href {https://doi.org/10.2168/LMCS-9(1:5)2013}
  {\path{doi:10.2168/LMCS-9(1:5)2013}}.

\bibitem{Leroux21}
J{\'{e}}r{\^{o}}me Leroux.
\newblock {T}he {R}eachability {P}roblem for {P}etri {N}ets is {N}ot
  {P}rimitive {R}ecursive.
\newblock In {\em 62nd {IEEE} Annual Symposium on Foundations of Computer
  Science, {FOCS} 2021, Denver, CO, USA, February 7-10, 2022}, pages
  1241--1252. {IEEE}, 2021.
\newblock \href {https://doi.org/10.1109/FOCS52979.2021.00121}
  {\path{doi:10.1109/FOCS52979.2021.00121}}.

\bibitem{LerouxS19}
J{\'{e}}r{\^{o}}me Leroux and Sylvain Schmitz.
\newblock {R}eachability in {V}ector {A}ddition {S}ystems is
  {P}rimitive-{R}ecursive in {F}ixed {D}imension.
\newblock In {\em 34th Annual {ACM/IEEE} Symposium on Logic in Computer
  Science, {LICS} 2019, Vancouver, BC, Canada, June 24-27, 2019}, pages 1--13.
  {IEEE}, 2019.
\newblock \href {https://doi.org/10.1109/LICS.2019.8785796}
  {\path{doi:10.1109/LICS.2019.8785796}}.

\bibitem{LincolnVWW18}
Andrea Lincoln, Virginia~Vassilevska Williams, and R.~Ryan Williams.
\newblock Tight hardness for shortest cycles and paths in sparse graphs.
\newblock In Artur Czumaj, editor, {\em Proceedings of the Twenty-Ninth Annual
  {ACM-SIAM} Symposium on Discrete Algorithms, {SODA} 2018, New Orleans, LA,
  USA, January 7-10, 2018}, pages 1236--1252. {SIAM}, 2018.
\newblock \href {https://doi.org/10.1137/1.9781611975031.80}
  {\path{doi:10.1137/1.9781611975031.80}}.

\bibitem{Lipton76}
Richard Lipton.
\newblock {T}he {R}eachability {P}roblem {R}equires {E}xponential {S}pace.
\newblock {\em Department of Computer Science. Yale University}, 62, 1976.

\bibitem{LokshtanonMS13}
Daniel Lokshtanov, D{\'a}niel Marx, and Saket Saurabh.
\newblock Lower bounds based on the {E}xponential {T}ime {H}ypothesis.
\newblock {\em Bulletin of EATCS}, 3(105), 2013.

\bibitem{Mayr81}
Ernst~W. Mayr.
\newblock {A}n {A}lgorithm for the {G}eneral {P}etri {N}et {R}eachability
  {P}roblem.
\newblock {\em {SIAM} J. Comput.}, 13(3):441--460, 1984.
\newblock \href {https://doi.org/10.1137/0213029} {\path{doi:10.1137/0213029}}.

\bibitem{MayrM82}
Ernst~W. Mayr and Albert~R. Meyer.
\newblock The complexity of the word problems for commutative semigroups and
  polynomial ideals.
\newblock {\em Advances in Mathematics}, 46(3):305--329, 1982.
\newblock URL:
  \url{https://www.sciencedirect.com/science/article/pii/0001870882900482},
  \href {https://doi.org/https://doi.org/10.1016/0001-8708(82)90048-2}
  {\path{doi:https://doi.org/10.1016/0001-8708(82)90048-2}}.

\bibitem{MazowieckiP19}
Filip Mazowiecki and Micha\l{} Pilipczuk.
\newblock {R}eachability for {B}ounded {B}ranching {VASS}.
\newblock In Wan~J. Fokkink and Rob van Glabbeek, editors, {\em 30th
  International Conference on Concurrency Theory, {CONCUR} 2019, August 27-30,
  2019, Amsterdam, the Netherlands}, volume 140 of {\em LIPIcs}, pages
  28:1--28:13. Schloss Dagstuhl - Leibniz-Zentrum f{\"{u}}r Informatik, 2019.
\newblock \href {https://doi.org/10.4230/LIPIcs.CONCUR.2019.28}
  {\path{doi:10.4230/LIPIcs.CONCUR.2019.28}}.

\bibitem{MazowickiSW23}
Filip Mazowiecki, Henry Sinclair{-}Banks, and Karol W\k{e}grzycki.
\newblock Coverability in 2-{VASS} with {O}ne {U}nary {C}ounter is in {NP}.
\newblock In Orna Kupferman and Pawe\l{} Soboci\'nski, editors, {\em
  Foundations of Software Science and Computation Structures}, pages 196--217.
  Springer Nature Switzerland, 2023.
\newblock \href {https://doi.org/10.1007/978-3-031-30829-1_10}
  {\path{doi:10.1007/978-3-031-30829-1_10}}.

\bibitem{Minsky67}
Marvin~L. Minsky.
\newblock {\em Computation: Finite and Infinite Machines}.
\newblock Prentice-Hall, Inc., 1967.

\bibitem{NesetrilP85}
Jaroslav Ne{\v{s}}et{\v{r}}il and Svatopluk Poljak.
\newblock On the complexity of the subgraph problem.
\newblock {\em Commentationes Mathematicae Universitatis Carolinae},
  26(2):415--419, 1985.

\bibitem{Rackoff78}
Charles Rackoff.
\newblock {T}he {C}overing and {B}oundedness {P}roblems for {V}ector {A}ddition
  {S}ystems.
\newblock {\em Theor. Comput. Sci.}, 6:223--231, 1978.
\newblock \href {https://doi.org/10.1016/0304-3975(78)90036-1}
  {\path{doi:10.1016/0304-3975(78)90036-1}}.

\bibitem{RosierY86}
Louis~E. Rosier and Hsu{-}Chun Yen.
\newblock {A Multiparameter Analysis of the Boundedness Problem for Vector
  Addition Systems}.
\newblock {\em J. Comput. Syst. Sci.}, 32(1):105--135, 1986.
\newblock \href {https://doi.org/10.1016/0022-0000(86)90006-1}
  {\path{doi:10.1016/0022-0000(86)90006-1}}.

\bibitem{Schmitz16}
Sylvain Schmitz.
\newblock {T}he {C}omplexity of {R}eachability in {V}ector {A}ddition
  {S}ystems.
\newblock {\em {ACM} {SIGLOG} News}, 3(1):4--21, 2016.
\newblock URL: \url{https://dl.acm.org/citation.cfm?id=2893585}.

\bibitem{ValiantP75}
Leslie~G. Valiant and Mike Paterson.
\newblock {D}eterministic {O}ne-{C}ounter {A}utomata.
\newblock {\em J. Comput. Syst. Sci.}, 10(3):340--350, 1975.
\newblock \href {https://doi.org/10.1016/S0022-0000(75)80005-5}
  {\path{doi:10.1016/S0022-0000(75)80005-5}}.

\bibitem{Aalst97}
Wil M.~P. van~der Aalst.
\newblock Verification of {W}orkflow {N}ets.
\newblock In Pierre Az{\'{e}}ma and Gianfranco Balbo, editors, {\em Application
  and Theory of Petri Nets 1997, 18th International Conference, {ICATPN} '97,
  Toulouse, France, June 23-27, 1997, Proceedings}, volume 1248 of {\em Lecture
  Notes in Computer Science}, pages 407--426. Springer, 1997.
\newblock \href {https://doi.org/10.1007/3-540-63139-9\_48}
  {\path{doi:10.1007/3-540-63139-9\_48}}.

\bibitem{YusterZ04}
Raphael Yuster and Uri Zwick.
\newblock Detecting short directed cycles using rectangular matrix
  multiplication and dynamic programming.
\newblock In J.~Ian Munro, editor, {\em Proceedings of the Fifteenth Annual
  {ACM-SIAM} Symposium on Discrete Algorithms, {SODA} 2004, New Orleans,
  Louisiana, USA, January 11-14, 2004}, pages 254--260. {SIAM}, 2004.
\newblock URL: \url{http://dl.acm.org/citation.cfm?id=982792.982828}.

\bibitem{Zagier79}
Don Zagier.
\newblock Newman's short proof of the prime number theorem.
\newblock {\em The American mathematical monthly}, 104(8):705--708, 1997.

\end{thebibliography}


\end{document}